\documentclass[aps,superscriptaddress,twocolumn,twoside,floatfix,pra,nofootinbib,a4paper]{revtex4-2}

\usepackage{main}
\usepackage{multirow}
\usepackage{makecell}

\begin{document}
\renewcommand{\ref}[1]{\text{\color{fr}Use {\tt\textbackslash{}cref}}}
\renewcommand{\lambda}{\text{\color{fr}Use {\tt\textbackslash saaa}, {\tt\textbackslash saab}, or {\tt\textbackslash sabc}. See {\tt packages/maths.sty}.}}

\title{Exploring the Local Landscape in the Triangle Network}

\author{Elisa B\"aumer}
\thanks{These authors contributed equally.}
\affiliation{Institute for Theoretical Physics, ETH Zurich, 8093 Z\"urich, Switzerland}

\author{Victor Gitton}
\thanks{These authors contributed equally.}
\affiliation{Institute for Theoretical Physics, ETH Zurich, 8093 Z\"urich, Switzerland}

\author{Tam\'as Kriv\'achy}
\thanks{These authors contributed equally.}
\affiliation{Atominstitut, Technische Universit\"at Wien, 1020 Vienna, Austria}
\affiliation{D\'epartement de Physique Appliqu\'ee, Universit\'e de Gen\`eve, CH-1211 Gen\`eve, Switzerland}

\author{Nicolas Gisin}
\affiliation{D\'epartement de Physique Appliqu\'ee, Universit\'e de Gen\`eve, CH-1211 Gen\`eve, Switzerland}
\affiliation{Schaffhausen Institute of Technology -- SIT, Geneva, Switzerland}

\author{Renato Renner}
\affiliation{Institute for Theoretical Physics, ETH Zurich, 8093 Z\"urich, Switzerland}

\begin{abstract}

Characterizing the set of distributions that can be realized in the triangle network is a notoriously difficult problem.
In this work, we investigate inner approximations of the set of local (classical) distributions of the triangle network.
A quantum distribution that appears to be nonlocal is the Elegant Joint Measurement (EJM) [Entropy. 2019; 21(3):325], which motivates us to study distributions having the same symmetries as the EJM.
We compare analytical and neural-network-based inner approximations and find a remarkable agreement between the two methods.
Using neural network tools, we also conjecture network Bell inequalities that give a trade-off between the levels of correlation and symmetry that a local distribution may feature.
Our results considerably strengthen the conjecture that the EJM is nonlocal.

\end{abstract}

\maketitle

\section{Introduction}

The problem of nonlocality, in general, requires to study classical (or ``hidden-variable'') causal models.
While this study is already extensive in the Bell scenario, featuring Alice and Bob sharing a common source, the case of more general networks remains poorly understood, with many open questions \cite{tavakoli_bell_2022}.
From a mathematical and computational perspective, the presence of multiple independent sources is the root of the difficulties associated to proving network nonlocality.
An interesting feature of network nonlocality is that, in general, and contrarily to Bell nonlocality, it is not necessary to provide different inputs to the parties to obtain a nonlocal behavior.
Here, we consider one of the simplest non-trivial networks, the triangle network, which consists of Alice, Bob and Charlie sharing only independent bipartite sources as depicted in \cref{fig:triangle}.
The first known example of quantum nonlocality in the triangle network dates back to the Fritz distribution \cite{Fritz2012}, which was shown to admit an experimental implementation that was provably nonlocal~\cite{polino_experimental_2023}.
This example of triangle nonlocality is quite specific, since the Fritz distribution is based on embedding a bipartite Bell test in the triangle network.
The other example of triangle nonlocality is the RGB4 family of distributions introduced in \cite{Renou2019} and later extended in \cite{abiuso_single-photon_2022,pozas-kerstjens_proofs_2023,boreiri_minimal_2023,boreiri_noise-robust_2023}.

\begin{figure}[tb!]
    \centering
        \includegraphics{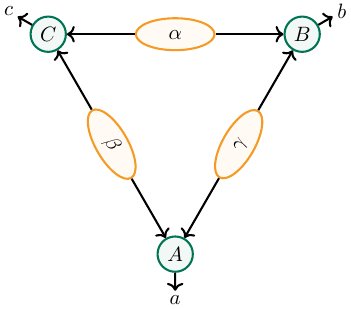}
	\caption{The triangle network features three observers (green circles), connected by three independent bipartite sources (yellow ellipses): Alice has access to the $\beta,\gamma$ sources, while Bob has access to $\gamma,\alpha$ and Charlie to $\alpha,\beta$.
        }
	\label{fig:triangle}
\end{figure}

Overall, compared to the Bell scenario, very few examples of triangle nonlocality are known to date.
In an attempt to find another example of nonlocality in the triangle network, the so-called \emph{Elegant Joint Measurement} (EJM) was introduced in \cite{gisin_published_ejm}.
This measurement is a two-qubit measurement that projects onto a partially entangled basis with a tetrahedral symmetry in the Bloch sphere for the marginal basis states.
The EJM can be used in different networks: in \cite{tavakoli2021}, a network Bell inequality has been tailored to be violated with the EJM applied in the bilocal scenario. This bilocal scenario can be thought of as the triangle network with the source between Alice and Bob being removed. The latter Bell inequality has been experimentally violated in \cite{Baeumer2021}.
When this measurement is used in the triangle network with the three sources prepared in singlet states, the resulting outcome distribution is referred to as the EJM distribution.
This EJM distribution exhibits high correlations between the three parties, as well as symmetry under both permutation of the parties and joint permutation of the outcomes.
The high level of symmetry of the distribution coupled to the relatively low number of outputs per party, namely, four, makes this distribution a prime candidate for the study of triangle nonlocality.
Indeed, the simplicity of the distribution suggests that a simple proof of nonlocality should exist, although this remains an open problem.
This is for instance in analogy to the case of Werner states \cite{werner_states}: their high level of symmetry is what enabled to show that there exist entangled states that admit a Bell-type hidden variable model.
The EJM distribution has been initially conjectured to be nonlocal~\cite{gisin_published_ejm}, and even conjectured to be nonlocal under reasonable experimental noise~\cite{krivachy_neural_2020}.
The interest for this distribution has been reinforced by a recent experimental implementation of the EJM measurement in the triangle network~\cite{wang2024experimental}.
Motivated by the open question of the nonlocality of the EJM, we consider the following problem: can we give a characterization of the local set in the triangle network for those distributions that share the same symmetries as the EJM distribution?

Our results are organized as follows.
In \cref{sec:setup}, we set the stage by defining triangle (non)locality, and we describe in particular the representation that we will use to depict a classical causal model in the triangle network, i.e., a classical strategy for Alice, Bob and Charlie.
We discuss quantum strategies and in particular the EJM.
We then introduce the notion of \emph{fully-symmetric} distributions, which refers to the distributions sharing the same symmetries as the EJM distribution.
We then move on to obtaining inner approximations of the set of fully-symmetric local distributions.
In \cref{sec:analytical flags}, we describe how we obtained a ``large'' analytical inner approximation.
We then describe in \cref{sec:neural} how a neural-network tool~\cite{krivachy_neural_2020} can be used to obtain another sort of inner approximation.
As we will see, the two methods agree surprisingly well, thus suggesting the quality and completeness of either method.
Moreover, aided with the neural network approach, we conjecture network Bell inequalities which would hold for any local distribution, even a non-symmetric one, and which are violated by the EJM distribution.
Our results suggest that the EJM distribution is relatively far from the local set, which would make it a prime target for a noise-robust proof of genuine triangle nonlocality.
We conclude in \cref{sec:conclusion} with remaining open problems and future research directions.

\section{Setup: locality and symmetry}
\label{sec:setup}

The setup we consider here is the triangle scenario, which is a network of three observers, Alice, Bob and Charlie, and three independent sources $\alpha$, $\beta$ and $\gamma$, as depicted in \cref{fig:triangle}.

\subsection{Classical setup}
\label{sec:flags}

\subsubsection{Local strategies}

If the sources are classical, then Alice, Bob and Charlie would use a classical strategy, i.e., a conditional distribution $p_A(a|\beta,\gamma)$ for Alice, $p_B(b|\gamma,\alpha)$ for Bob and $p_C(c|\alpha,\beta)$ for Charlie.
We say that a distribution $p(a,b,c)$ is \emph{local} or classical\footnote{A more precise terminology would be to call such a distribution classically local in the triangle network. This would be in contrast to \emph{quantumly} local distributions in the triangle network, where ``local'' would simply refer to the connectivity of the triangle network. However, following the standard terminology of the field, we use the abbreviation ``local distribution'' in place of ``classically local distribution''.} (in the triangle network) if there exist $p_A$, $p_B$ and $p_C$ such that
\begin{multline}
\label{trilocal}
	p(a,b,c) = \\
	\int_{[0,1]^{\times3}} \dd \alpha \dd \beta \dd \gamma  p_A(a|\beta, \gamma) \, p_B(b|\gamma,\alpha) \, p_C(c|\alpha,\beta),
\end{multline}
where $\alpha,\beta,\gamma \in [0,1]$ represent the values distributed by each source.
Otherwise, $p(a,b,c)$ is said to be \emph{nonlocal}.
The set of local distributions is non-convex, which makes its description quite complex.
Note that we assumed without loss of generality that the sources in this classical setting are distributed uniformly over the interval $[0,1]$.
We may also assume without loss of generality that the local strategies determining the output of each party are given by deterministic response functions $a(\beta,\gamma)$, $b(\alpha,\gamma)$ and $c(\alpha,\beta)$, respectively.
The probabilistic response functions are then related to the deterministic response functions by
\begin{align*}
    p_A(a|\beta,\gamma) &= \delta_{a,a(\beta,\gamma)}, \\
    p_B(b|\gamma,\alpha) &= \delta_{b,b(\gamma,\alpha)}, \\
    p_C(c|\alpha,\beta) &= \delta_{c,c(\alpha,\beta)}.
\end{align*}
As every function depends on only two variables, we can nicely illustrate the two independent parameters for each party as two dimensions and indicate the corresponding outputs by different colors. 
This maps the response functions to what we call ``flags'', that might give a more intuitive understanding of the different strategies. 
Let us illustrate that with the following example.

\subsubsection{Example flags for a highly correlated distribution}
\label{sec:highcorr}

 \begin{figure}[htb]
	\centering
        \includegraphics{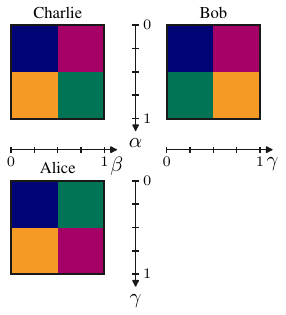}
	\caption{Flags illustrating the response functions that the three parties Alice, Bob and Charlie may use in the triangle network (\cref{fig:triangle}).
    For instance, Alice outputs her outcome $a$ given her two input values $\beta,\gamma \in [0,1]$.
    The different values of the outcomes (such as $a$) are indicated by different colors. }
	\label{fig:highcorr}
\end{figure}
\Cref{fig:highcorr} showcases output strategies of the three parties Alice, Bob and Charlie, as functions of their respective local variables. 
Let us denote {\color{fb}``blue''} by $1$, {\color{fr}``pink''} by $2$, {\color{fy}``yellow''} by $3$ and {\color{fg}``green''} by $4$. Then for example Charlie's strategy is given by
\begin{itemize}
	\item output {\color{fb}``blue''} $\equiv 1$ for $\alpha \in [0,0.5]$, $\beta \in [0,0.5]$,
	\item output {\color{fr}``pink''} $\equiv 2$ for $\alpha \in [0,0.5]$, $\beta \in (0.5,1]$,
	\item output {\color{fy}``yellow''} $\equiv 3$ for $\alpha \in (0.5,1]$, $\beta \in [0,0.5]$,
	\item output {\color{fg}``green''} $\equiv 4$ for $\alpha \in (0.5,1]$, $\beta \in (0.5,1]$.
\end{itemize}
The strategies for Alice and Bob can be determined analogously.
In this example, the output probability distribution can be summarized as follows:
\begin{align}
\label{eq:high111local}
	p(a,b,c) &= \frac{1}{8} \big( [1,1,1] + [2,2,2] + [3,3,3] + [4,4,4] \nonumber\\
	&\quad + [1,4,3]+[2,3,4] + [3,2,1] + [4,1,2] \big).
\end{align}
where we defined the deterministic distribution $[k,l,m]$ as $[k,l,m](a,b,c) = \delta_{ak}\delta_{bl}\delta_{cm}$.
Note that while here the probability that all three parties output the same value is given by 
$p(A=B=C) = \sfrac{1}{2}$,
indicating very strong correlations, the distribution is not symmetric.
In \cref{sec:analytical flags} we use this flag model to investigate local strategies that yield fully symmetric probability distributions.

\subsection{Quantum setup}
\label{sec:quantumsetup}

\subsubsection{Quantum strategies}

Let us move on to the quantum setting, where instead of outputting classical values that follow a classical probability distribution, the sources distribute entangled quantum systems.
Note that in this setting, contrary to the standard Bell nonlocality tests, the observers receive no setting, or input, other than their respective two quantum systems.
Quantum outcome distributions include as a special case the local distributions of \cref{trilocal}, but can also create nonlocal distributions~\cite{Fritz2012,Renou2019}.
However, in general, it is not straightforward to create such nonlocal distributions or to demonstrate nonlocality for a given distribution, i.e., to prove that the distribution could not be reproduced by a classical local model.

In \cite{Renou2019}, nonlocal quantum distributions were found in the triangle network by using quantum sources and joint measurements with entangled eigenstates. 
The nonlocality of these distributions could not be traced back to the standard violation of Bell inequalities. 
Thus, their nonlocality appears to be fundamentally different, which is a major step toward characterizing true quantum phenomena. 
However, no reasonable noise-robust proof of nonlocality has yet been found, rendering an experimental implementation impossible. 
In the next section, we present in more detail another entangled measurement scheme, which is conjectured to be nonlocal with an appropriate noise-robustness. 
In addition to high correlations, it features a very high level of symmetry.

\subsubsection{The Elegant Joint Measurement}
\label{sec:EJM}

The Elegant Joint Measurement (EJM) was first introduced in 2017 and describes a measurement of two qubits projected onto a basis of partially entangled states with a tetrahedral symmetry (see \cite{gisin_published_ejm} for more details).
When applied to the triangle scenario, the setting considered in \cite{gisin_published_ejm} starts with all three parties sharing pairwise the maximally entangled singlet state $\ket{\Psi^{-}} = \frac{1}{\sqrt{2}} (\ket{01}-\ket{10})$ and then performing the EJM onto their two respective qubits. 
Each party obtains an output $a,b,c\in\{1,2,3,4\}$, respectively, and as the resulting probability distribution is highly symmetric, it can be fully described by only three cases: all outcomes are equal,
exactly two outcomes are equal,
or all outcomes are different.
The distribution can be described as follows: for all $a,b,c\in\{1,2,3,4\}$,
\begin{equation*}
p(a,b,c) = \left\{\begin{aligned}
    \displaystyle\frac{25}{256} &\hspace{10pt}\text{if } a=b=c, \\
    \displaystyle\frac{5}{256} &\hspace{10pt}\text{if } a\neq b \neq c\neq a, \\
    \displaystyle\frac{1}{256} &\hspace{10pt}\text{else.} 
\end{aligned}\right.
\end{equation*}
Although this distribution has strong correlations, i.e., a large probability ${p(A=B=C)}$, one could also find classical models with even higher correlations (see the example in \cref{sec:highcorr}). 
It is however conjectured that this specific distribution is nonlocal due to its additional high degree of symmetry~\cite{EJM2019}, i.e., that the distribution cannot be written as \cref{trilocal}.
In addition, unlike previous nonlocal quantum correlations, the EJM's nonlocality is also conjectured to be noise-robust~\cite{krivachy_neural_2020}.
This served as another motivation to investigate fully symmetric distributions in the triangle network.

\subsection{Fully symmetric distributions}
\label{sec:fully symmetric distributions}

We now focus on the case of four outcomes per party, since this is the case for the EJM distribution of interest.
Let us first define what we mean by \textit{fully symmetric distributions}: 
a distribution $p$ if fully symmetric if it is symmetric under permutation of the parties as well as under joint permutation of the outcomes (note that distributions which are symmetric under joint permutations of the outcomes were referred to as output-permutation-invariant (OPI) distributions in~\cite{girardin_violation_2023}).
This implies that a fully symmetric distribution $p$ can be characterized by only three values, 
\begin{align}
\label{eq:def s111}
	\saaa &= p(A=B=C) 
        = 4\,p(1,1,1),\nonumber\\
	\saab &= p(A=B\neq C) + p(A=C\neq B) + p(B=C\neq A) \nonumber\\
        &= 36\,p(1,1,2), \nonumber\\
	\sabc &= p(A\neq B\neq C \neq A)
        = 24\,p(1,2,3).
\end{align}
In fact, the normalization $\saaa + \saab + \sabc = 1$ implies that only two values could be used.
Additionally, let us define the three extremal fully symmetric distributions, $\paaa$, $\paab$ and $\pabc$ as follows:
\begin{align*}
    \paaa &= \frac{1}{4}\sum_{k} [k,k,k], \\
    \paab &= \frac{1}{36}\sum_{k\neq l} [k,k,l] + [k,l,k] + [l,k,k], \\
    \pabc &= \frac{1}{24}\sum_{k\neq l\neq m\neq k} [k,l,m].
\end{align*}
With the above definitions, we have that any fully symmetric distribution $p$ can be written as
\begin{align*}
    p = \saaa \paaa + \saab \paab + \sabc \pabc,
\end{align*}
or in other words, for all $a,b,c\in\{1,2,3,4\}$,
\begin{multline*}
    p(a,b,c) = \saaa\paaa(a,b,c) + \saab\paab(a,b,c) \\
    + \sabc\pabc(a,b,c).
\end{multline*}
For instance, the EJM (\cref{sec:EJM}) is characterized by
\begin{align*}
    (\saaa,\saab,\sabc) = \left(\frac{25}{64}, \frac{9}{64}, \frac{30}{64} \right).
\end{align*}
The Finner inequality \cite{Finner1992} states that both local and quantum strategies in the triangle network lead to distributions satisfying
\begin{align}
\label{eq:general finner}
    p(a,b,c) \leq \sqrt{p(A=a)p(B=b)p(C=c)}.
\end{align}
In the special case of fully symmetric distributions, the marginals are maximally mixed, since by symmetry we must have $p(A=a) = p(A=\sigma(a))$ for all $\sigma \in S_4$, and similarly for the other marginals.
Hence, \cref{eq:general finner} simplifies to
\begin{align*}
    p(a,b,c) \leq \frac{1}{8},
\end{align*}
which implies nothing for $\saab$ and $\sabc$ (since it only states that $\saab \leq \sfrac{36}{8}$ and $\sabc \leq \sfrac{24}{8}$), but it implies that
\begin{align}
\label{eq:finner four outcomes}
    \saaa \leq \frac{1}{2}
\end{align}
for all local and quantum distributions~\cite{Renou_quantumFinner_2019}. 
Note that the Finner inequality may not be valid in physical theories that merely satisfy the non-signaling principle, as investigated in~\cite{girardin_violation_2023}.

\section{Analytical construction of local models}
\label{sec:analytical flags}

In this section, we describe analytical local models leading to fully symmetric distributions.
In the context of trying to find a local model for the EJM distribution, a similar kind of parametrized local model was proposed in~\cite{gisin_published_ejm}.
The resulting local distributions are fully symmetric: the results of~\cite{gisin_published_ejm} prove that for all $t \in [0,1]$, the fully symmetric distribution described by
$$
    (\saaa,\saab,\sabc) = \left( \frac{52 + 9t}{256}, \frac{180 + 9t}{256}, \frac{24 - 18t}{256}
    \right)
$$
is local.
This family of distributions can be seen as a special case of the more general constructions that we describe in the following.

\subsection{Description of the constructions}

To obtain an inner approximation of the set of local distributions within the symmetric subspace, we construct analytical local models that give rise to a fully symmetric distribution $p(A,B,C)$.
We do so using the flag model introduced in \cref{sec:flags}.
The first step in doing so was to devise a method for generating flags that yield an outcome-symmetric distribution.
This method is described in \cref{sec:defining outcome-symmetric local models,sec:generating outcome-symmetric local models}.
It essentially relies on ensuring that outcome permutations can be ``cancelled'' by a suitable permutation of the values of the three classical sources $\alpha,\beta,\gamma$.
The flags that we construct typically come with a few analytical parameters, such as $q$ and $\nu$ in \cref{fig:symmmaxcorr}.
We could then compute the output distribution of such flags in terms of those parameters, and enforce party symmetry by imposing suitable relations between those parameters.
The explicit constructions are described in \cref{sec:high_SYM_corr,app:genflags,app:anticorrflags}.
As an example, we depict in \cref{fig:symmmaxcorr} flags that yield a highly correlated and fully symmetric distribution (upon imposing the appropriate relation on $q$ and $\nu$, see \cref{sec:high_SYM_corr}).

\begin{figure}[htb]
    \centering
    \includegraphics{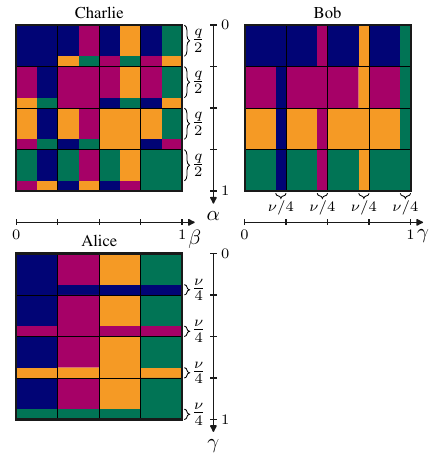}
    \caption{
        Flags illustrating response functions $a(\beta,\gamma)$, $b(\alpha,\gamma)$ and $c(\alpha,\beta)$ that yield the maximal three-party correlation $p(A=B=C)=\sfrac{1}{4}$ that a fully-symmetric local distribution can achieve within the analytical inner approximation that we considered.
        The flags are parameterized by $q \in [0,\sfrac12]$ and $\nu\in[0,1]$, and they give rise to the distribution of \cref{generallambda} with $r = \eta = 1$ if $\nu\in[0,\sfrac13]$ and $q = \sfrac{\nu}{(1-\nu)}$.
        See \cref{sec:high_SYM_corr} for more details on this construction.
    }
    \label{fig:symmmaxcorr}
\end{figure}

The most general family of flags that we found, described in \cref{app:genflags}, comes with three parameters: $r,\eta \in [0,1]$ and $\nu \in [0,\sfrac12]$, satisfying the relation
$$
    0 \leq \frac{1-r}{3}+\frac{\nu}{1-\nu} \frac{4 \eta -1}{3} \leq \frac12.
$$
They yield the following family of fully symmetric distributions:
\begin{align}
	\saaa &= \frac{1}{4} \big((1-\nu)r+\eta \nu\big) \nonumber\\
	\saab &= \frac{3}{4} \big((1-\nu) (1-r) + 3 \eta \nu\big) \nonumber\\
	\sabc &= \frac{1}{4} \big( 1+(1-\nu) 2r + (3-10 \eta )\nu\big).
	\label{generallambda}
\end{align}

In a different construction, presented in \cref{app:anticorrflags}, we allowed Alice, Bob and Charlie to further anti-correlate their strategies, achieving distributions outside of the family of distributions described above.
This results in a line that can be described by
\begin{align}
(\saaa, \saab, \sabc)= \Big(\frac{r}{48}, \frac{4-r}{16},\frac{18+r}{24}\Big),
\label{eq:bottomrightline}
\end{align}
for $r\in[0,1]$.
As described in \cref{app:currents}, we can extend this line of feasible distributions into a more general family of distributions by applying a generic ``decorrelation'' procedure in which the players and the sources modify their behaviors with a certain probability: this yields the little two-dimensional ``spike'' at the bottom right of \cref{fig:triangleanalytic}.

\subsection{Visualizing the local set}

To visualize the set of fully symmetric local models, we can plot all combinations $(\saaa, \saab, \sabc)$ in a triangle, where the three corners correspond to the three extremal points $\saaa=1$ (top), $\saab=1$ (bottom left) and $\sabc=1$ (bottom right), respectively. 
All other points are convex combinations and lie inside of the triangle, or on an edge if one of the coefficient $(\saaa,\saab,\sabc)$ is $0$.
In \cref{table:lines}, we describe some segments of local distributions in the fully symmetric subspace that can be obtained as special cases of \cref{generallambda}: some of these visibly lie on the boundary of the local region implied by \cref{generallambda}.
The local distributions of \cref{generallambda,eq:bottomrightline} are plotted in \cref{fig:triangleanalytic}.

\newcommand{\mybegline}{\rule{0pt}{16pt}}
\newcommand{\myendline}{\\[7pt]}
\newcommand{\writepoint}[1]{$\displaystyle\left(#1\right)$}
\begin{table}[ht]
\begin{tabular}{|c|cc|c|}
\hline
{\multirow{2}{*}{Fixed values}} & \multicolumn{2}{c|}{Line $(\saaa,\saab,\sabc)$} & Color in \\
 & From & To & \cref{fig:triangleanalytic} \\
\hline\hline
\mybegline $\eta = 1, \ r=1$ & \writepoint{\frac14,\frac34,0} & \writepoint{\frac{1}{4},0,\frac{3}{4}} & \textcolor{fr}{purple} \myendline
\hline
\mybegline $\eta = 1, r= \displaystyle\frac{7\nu-1}{2(1-\nu)}$ & \writepoint{\frac{1}{28},\frac{27}{28},0} & \writepoint{\frac{1}{4},\frac{3}{4},0} & \textcolor{red}{red} \myendline
\hline
\mybegline $\eta = 1, \ r=0$ & \writepoint{0,\frac{3}{4},\frac{1}{4}} & \writepoint{\frac{1}{28},\frac{27}{28},0} & \textcolor{black!80}{grey} \myendline
\hline
\mybegline $\eta = 0, \displaystyle r=\frac{1-2\nu}{1-\nu}$ & \writepoint{0,\frac{3}{8},\frac{5}{8}} & \writepoint{\frac{1}{4},0,\frac{3}{4}} & \textcolor{fg}{dark green} \myendline
\hline
\mybegline $\eta = 0, \ r=0$ & \writepoint{0,\frac{3}{4},\frac{1}{4}} & \writepoint{0,\frac{3}{8},\frac{5}{8}} & \textcolor{blue}{light blue} \myendline
\hline
\mybegline $\nu = 0$ & \writepoint{0,\frac{3}{4},\frac{1}{4}} & \writepoint{\frac14,0,\frac{3}{4}} & \textcolor{fb}{dark blue} \myendline
\hline
\end{tabular}
\caption{We fix some of the parameters of \cref{generallambda} to obtain a few key lines, which are also represented in \cref{fig:triangleanalytic}.}
\label{table:lines}
\end{table}

\begin{figure}[htb]
	\centering
        \includegraphics{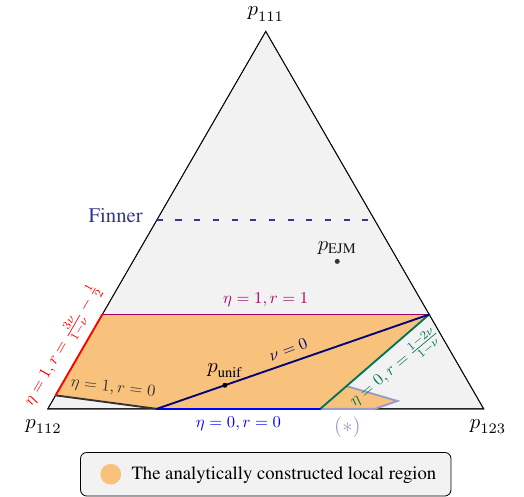}
	\caption{The set of fully symmetric local distributions given by \cref{generallambda}.
    The respective lines are indicated explicitly in \cref{table:lines}. Additionally, we plot the construction described around \cref{eq:bottomrightline}: it is indicated by the \textcolor{fb!40}{$(*)$} marker.
    We also indicate the Finner inequality, see \cref{eq:finner four outcomes}.
    The fact that the EJM distribution $\pejm$ (see \cref{sec:EJM}) is quite far from all the local distributions we found supports the conjecture that it is nonlocal.}
	\label{fig:triangleanalytic}
\end{figure}

From the construction in \cref{app:genflags} we can see that the cases with $\eta = 1$, which are at the top left of the set, describe the strategies where Alice and Bob use their common source to maximally correlate, while for the cases with $\eta=0$, which are on the bottom right of the set, they use it to maximally anti-correlate.
Note that the case $\eta = 1$, $r= 1$, which corresponds to a horizontal line at $\saaa = \sfrac{1}{4}$, corresponds to the maximally correlated flags constructed in \cref{sec:high_SYM_corr} and shown in \cref{fig:symmmaxcorr}.
The case of $\nu=0$, dividing the cases with $\eta = 1$ and $\eta = 0$, describes the strategies where Alice and Bob do not use their common source at all. 
The resulting distributions could be obtained already in the bilocal network, which corresponds to the triangle network but without the source between Alice and Bob.
These include the maximally mixed distribution, where $p(a,b,c) = \sfrac{1}{64}$, yielding $(\saaa,\saab,\sabc) = (\sfrac{1}{16}, \sfrac{9}{16},\sfrac{6}{16})$, which can be seen as the most trivial symmetric strategy, as it can be reached by every party just randomly outputting each of the four outcomes with equal probability.

\section{Neural Network Approaches}
\label{sec:neural}
Numerical searches of local models in the triangle network are difficult, as the problem is non-convex and one often ends up in local minima. 
However, modeling local response functions with artificial neural networks has been shown to be a relatively reliable heuristic, reproducing benchmark results, as well as providing new conjectures, which have since been partially proven~\cite{krivachy_neural_2020,pozas-kerstjens_proofs_2023}. 
In general, a feed-forward artificial neural network is a numeric model for a multivariate, multidimensional function. It can be trained, i.e., its parameters can be fit, in order to minimize an objective function that depends on the neural network's outputs in a differentiable manner.

\subsection{Minimizing distance to target distributions}

For local models in networks, one can model each of the response functions in \cref{trilocal} with neural networks, i.e., the neural network for Alice would take as inputs some $\beta_i$ and $\gamma_i$ values and output (a normalized, not necessarily deterministic) $p_{A}^{\text{NN}}(a|\beta_i,\gamma_i)\in\mathbb{R}^4$, and similarly for Bob and Charlie (this architecture is often referred to as LHV-Net). 
An example of response functions that LHV-Net ``thought about'' is displayed in \cref{fig:nn flags}.
Sampling over $M\gg1$ triples $(\alpha_i,\beta_i,\gamma_i)\in[0,1]^{\times3}$, one can then numerically calculate a Monte Carlo estimate of \cref{trilocal},
\begin{equation*}
    p_{\text{NN}} = \frac{1}{M}\sum_{i=1}^{M} p_{A}^{\text{NN}}(a|\beta_i,\gamma_i) p_{B}^{\text{NN}}(b|\gamma_i, \alpha_i) p_{C}^{\text{NN}}(c|\alpha_i,\beta_i).
\end{equation*}
Crucially, the neural networks of Alice, Bob, and Charlie only have access to the respective hidden variables allowed by the triangle structure, thus \textit{any} distribution given by LHV-Net is by definition local. The three neural networks are then jointly optimized by minimizing the objective function $||p_{\text{target}}-p_{\text{NN}}||^2$. In our case, we use a multilayer perceptron of depth 4 and width 30 with rectified linear activation functions for each party, with an Adadelta optimizer, and stochastic gradient descent used for fine-tuning the weights. We designated 260 target distributions in the symmetric subspace, and trained the neural network 10 times independently for each of them. Finally, we kept the best resulting model for each point. 

\begin{figure}[tb]
    \centering
    \includegraphics{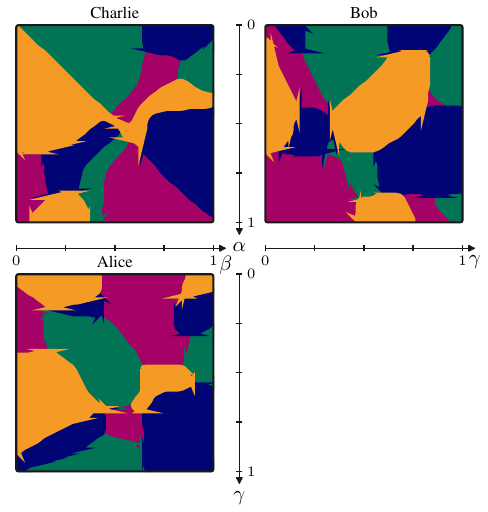}
    \caption{Illutration made from some of the response functions that the neural network came up with. To obtain this figure, we turned the probabilistic reponse functions into deterministic ones by picking the most likely outcome for each input values.
    See \cref{app:saaa>0.25} for more details about these specific flags.}
    \label{fig:nn flags}
\end{figure}

The results are displayed in \cref{fig:tamas}.
It is important to note that the neural network's output distributions are not forced to be symmetric, i.e., they are not necessarily actually within the fully symmetric subspace. However, those that are close in 2-norm to their respective (fully symmetric) targets are naturally close to being symmetric. 
Hence, the resulting dark blue region portrayed in \cref{fig:tamas} gives a good indication of a region that is local in the fully symmetric subspace.
Further plots of scans of the symmetric subspace for $N$ outcomes with $N=3,5,6$ can be found in \cref{app:generalN}.

\begin{figure}[htb]
	\centering
        \includegraphics{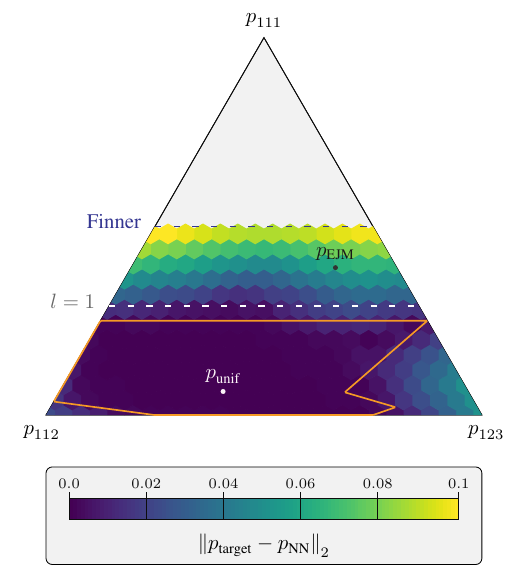}
	\caption{Results from the neural network trained to find local strategies yielding a distribution close to each point inside of the fully symmetric subspace. The color indicates the distance to the actual distribution, which should be very small (i.e., dark blue) in order to consider a local strategy as ``detected''. The orange line delimits the local region described in \cref{fig:triangleanalytic} and is drawn to emphasize the good agreement between the analytical local models we found and what the neural network perceives as local.
    We also indicate the Finner inequality, see \cref{eq:finner four outcomes}, as well as the $l=1$ conjectured Bell inequality of \cref{eq:ineq_l1}.
 }
	\label{fig:tamas}
\end{figure}

\subsection{Finding inequalities}
\label{sec:finding inequalities}

In general causal scenarios, Bell-type inequalities are those which are satisfied by correlations from local models, and hopefully violated by some quantum correlations. 
Obtaining such inequalities for the triangle network has proven difficult, with previous attempts not finding any genuine quantum violations or resulting in difficult-to-interpret inequalities~\cite{Henson_2014,Frase_2018,Weilenmann_2018_nonshannon}.

Regarding the symmetric subspace, our previous analytic and numeric findings give strong evidence that \textit{within} the symmetric subspace the $\saaa$ of local models is limited  by some value $\saaa^*$, with $\saaa^* \approx \sfrac14$, allowing for a simple inequality that would rule out the Elegant distribution. However, \textit{outside} this subspace, local models can reach high (in fact, maximal) $\saaa$ values. Formalizing this intuition (that it is difficult for local models to simultaneously have large $\saaa$ and be symmetric) as an inequality for generic distributions is a priori difficult. 

By changing the objective function of LHV-Net, we can test different ans\"atze for Bell-type inequalities for the triangle network. We do this by penalizing asymmetry with a penalty term on each of the three types of probabilities appearing in symmetric distributions (111, 112, 123). Specifically, we sum up the joint deviation from the mean for each of these types of probabilities, 
\begin{align*}
\Delta_{l} &= \sum_{X \in \{111,112,123\}} \Delta_{l,X},\\
\Delta_{l, X} &= \sum_{\{a,b,c\} \in \mathcal{I}_X} |M_X - p(a,b,c)|^l,\\
M_{X} &= \frac{1}{|\mathcal{I}_X|} \sum_{\{a,b,c\} \in \mathcal{I}_X} p(a,b,c),
\end{align*}
where $\mathcal{I}_X$ is the index set of $X$-type outcomes, and $l\in\{1,2\}$. In particular $\mathcal{I}_{111}$ will contain 4 elements, $\mathcal{I}_{112}$ 36, and $\mathcal{I}_{123}$ 24 elements.

Consider now the following function:
\begin{align*}
    f_w(p) = w \cdot \saaa(p) - (1-w) \Delta_{l}(p).
\end{align*}
By maximizing this quantity for LHV models, we can see whether it can outperform $f_w(\pejm)$ for any value of $w$. 
Intuitively, we are trying to maximize $\saaa$ (with weight $w$) and minimize the penalty $\Delta_{l}$ (with weight $1-w$). 
We define the gap
\begin{align*}
    \delta_w:=f_w(\pejm) - \max_{p\in\mathcal{L}} f_w(p),
\end{align*}
which, if positive, defines an inequality
\begin{align}
\label{eq:ineq:basic}
    f_w(p) \leq f_w(\pejm) - \delta_w,
\end{align}
which all local distributions must obey.

Finding the exact value of $\delta_w$ is difficult, as one must optimize over local models of the form of \cref{trilocal}. However, using LHV-Net, one can obtain an estimate of $\delta_w$, by setting the objective function to be $-f_w(p)$.
For each $w$ value, we train the neural network from scratch to try to violate that given inequality and plot the resulting $\delta_w$ values in \cref{app:inequality_params}, for both absolute value and square penalty ($l=1,2$). We find the largest $\delta_w$ for $l=1$ ($l=2$) at $w^*\approx 0.678$ ($w^*\approx 0.161$) with $\delta_{w^*}\approx 0.069$ ($\delta_{w^*}\approx 0.012$), getting that approximately
\begin{align}
\label{eq:ineq_l1}
    \saaa(p) - 0.475 \,\Delta_{l=1}(p) &\leq 0.289, \\
\label{eq:ineq_l2}
    \saaa(p) - 5.211 \,\Delta_{l=2}(p) &\leq 0.316,
\end{align}
should both hold for all local models. Recently, the inequality with $l=2$ has been violated for a range of $w$ values by experimentally obtained data~\cite{wang2024experimental}.
In the symmetric subspace, the penalty term vanishes, which means that \cref{eq:ineq_l1} is the most constraining of the two inequalities.
This is the inequality that we plotted in \cref{fig:tamas}.

Interestingly, when maximizing these inequalities, the neural network seems to find something better than the $\saaa=\sfrac{1}{4}$ strategy that was found analytically. 
This apparent out-performance of the $\sfrac14$ bound also appears in \cref{fig:tamas}, most prominently at values where $p(112) \approx p(123)$. 
In \cref{app:inequality_params} we include details about the local strategy that the neural network found with $\saaa\approx 0.289$ ($\Delta_{l=1} \approx 0.009$, $\Delta_{l=2} \approx 2.3 \cdot 10^{-6}$), and its discretized, deterministic approximation, which is displayed in \cref{fig:nn flags}, with $\saaa\approx 0.294$ ($\Delta_{l=1} \approx 0.014$, $\Delta_{l=2} \approx 4.7 \cdot 10^{-6}$). 
It is currently an open question whether there exists an exactly symmetric local distribution with $\saaa>\sfrac14$ (e.g.~$\saaa = 0.289$ as implied by the inequality), or whether these distributions only exist very close to the symmetric subspace. 
Note that in the case of 3 outcomes per party, it is possible to find a fully symmetric local distribution with $\saaa > \sfrac13$ as shown in \cref{app:3outcomesExample}.

Finally, one could try a variety of different penalty functions for constructing inequalities. 
However, several simple ones that do not work are using only $\Delta_{l} = \Delta_{l,111}$ or when imposing symmetry only at the level of the single-party marginals. 
Both of these types of penalty functions have the local distribution in \cref{eq:high111local} as an example of why such penalty functions would not work: this local distribution would get zero penalty, however, its $\saaa$ is larger than that of the EJM distribution.

\section{Conclusion}
\label{sec:conclusion}

The high level of symmetry in the correlations obtained from performing the Elegant Joint Measurement \cite{gisin_published_ejm} in the triangle network, which are conjectured to lead to noise-robust nonlocal quantum correlations, inspired us to investigate fully symmetric distributions in this setting, i.e., distributions that are symmetric under permutation of the parties and under joint permutations of the outcomes. 

We analytically constructed classical local model and applied neural network techniques to substantiate our findings.
The agreement between the two methods is best witnessed in \cref{fig:tamas}.
Both methods are fundamentally inner approximations, i.e., they can only certify that a given distribution \emph{is} local.
However, the good agreement between the two methods suggests that, in this case, these methods are essentially able to find a local model for a distribution if the distribution is local.
Of course, the exact location of the boundary of the local set is still hard to pinpoint, but in any case, this boundary seems fairly far away from the Elegant Joint Measurement distribution.
Moreover, we formalized the trade-off that local models face between being highly correlated and highly symmetric via the conjectured Bell inequalities (see \cref{sec:finding inequalities}).

This led to the conjecture that local models yielding a fully symmetric outcome distribution with four outcomes per party have a maximal correlation very close to $p(A=B=C)=\sfrac{1}{4}$. 
While the $\sfrac{1}{4}$ maximum value is what we found with our analytical construction, the neural network approach supports that indeed the correlation cannot be much higher, but also indicates that the true value might be slightly above $\sfrac{1}{4}$. 
\begin{openproblem*}
    Does there exist a distribution $p(A,B,C)$ that is local in the triangle network and fully symmetric such that $p(A=B=C) > \sfrac14$ ?
\end{openproblem*}
However, as the EJM is well above that bound, even a slightly higher upper bound would still imply that it leads to noise-robust nonlocal quantum correlations. 
For $N\geq3$ outcomes per party, we can construct strategies with $p(A=B=C)=\sfrac{1}{N}$, and let neural networks substantiate that the local upper bound must be close to $\sfrac1N$ (see \cref{app:generalN}).
However, in the $N=3$ outcome case, we found an analytical local model that gives rise to $\saaa > \sfrac13$ (see \cref{app:3outcomesExample}), but were not able to generalize this to $N\geq4$ outcomes.

It still remains open to find a proof that the Elegant Joint Measurement is nonlocal in the triangle network.
Ideally, it would be even better to obtain a more general bound on the maximal correlations for local models like our suggested conjecture, or an even more general network Bell inequality that can be used to identify nonlocality irrespectively of symmetry.
In parallel, it would be useful to develop similar inner approximation tools for the set of quantum distributions in the triangle network.
Such an exploration could in particular provide additional examples of quantum nonlocality on top of the few examples that are known today.

\section{Data availibility}
A data appendix is available at Ref.~\cite{data_appendix}. The following are included.
\begin{itemize}
    \item The ($w,\delta_w$) pairs that were found by LHV-Net for \cref{eq:ineq:basic} (both for $l=1$ and $l=2$).
    \item The almost symmetric distribution found by LHV-Net that has $\saaa>\sfrac14$, as well as the corresponding (discretized, deterministic) flags that generate it.
    \item Data used in the LHV-Net maps of the symmetric subspace for $N=3,4,5,6$ (\cref{fig:tamas,fig:nn cards}).
    \item Accompanying scripts to load and evaluate this data.
    \item A Wolfram Mathematica script that displays the different analytical flags that we used in \cref{sec:analytical flags}, and computes the associated outcome distributions.
    \item A Python script that displays the inner approximation of \cref{fig:triangleanalytic}.
\end{itemize}

\section{Acknowledgements}

This work was supported by the Swiss National Science Foundation (NCCR SwissMAP, as well as project No.\ 200021\_188541 and P1GEP2\_199676 (TK)) and by the Air Force Office of Scientific Research via grant FA9550-19-1-0202. TK was additionally funded by the European Research Council (Consolidator grant ’Cocoquest’ 101043705) and the Austrian Federal Ministry of Education via the Austrian Research Promotion Agency–FFG (flagship project FO999897481, funded by EU program NextGenerationEU). Most of the neural network computations were performed at University of Geneva on ``Baobab'' HPC cluster.



\bibliography{bibliography}


\appendix
\onecolumngrid

\section{General flag constructions}

In this section, we describe the general four-outcome flag models that underlie \cref{sec:analytical flags}.

\subsection{Two-party marginals of the Elegant Joint Measurement}

We start by noting that the Elegant Joint Measurement has local two-party marginals in the following sense.
In \cref{fig:twoparty}, we provide an explicit strategy such that all the two-party marginals are equal to those of the Elegant Joint Measurement, i.e.,
\begin{align*}
	p(A = k,B = l) &= p(A = k,C = l) = p(B = k,C = l)
	=
	\left\{\begin{aligned}
		\displaystyle\frac{7}{64} &\hspace{10pt} \text{if } k=l,\\
		\displaystyle \frac{3}{64} &\hspace{10pt} \text{if $k\neq l$}.
	\end{aligned}\right.
\end{align*}
However, the full distribution $p(A,B,C)$ obtained from this strategy is not equal to that of the Elegant Joint Measurement.
In particular, the three-party correlations are quite small, with
\begin{align}
    p(A = B = C = k) = \frac{1}{16}.
\end{align}
Additionally, we note that while the resulting three-party-distribution is symmetric under cyclic permutations of the parties, it is not symmetric under permutation of the outcomes.
\begin{figure}[htb]
	\centering
        \includegraphics{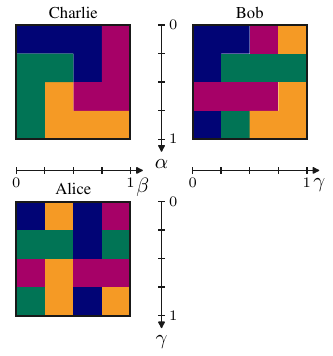}
	\caption{Flags illustrating response functions $a(\beta,\gamma)$, $b(\alpha,\gamma)$ and $c(\alpha,\beta)$ that yield the same two-party marginals as the EJM, but not the right three-party probabilities.
	}
	\label{fig:twoparty}
\end{figure}

\subsection{Defining outcome-symmetric local models}
\label{sec:defining outcome-symmetric local models}

As we would like to characterize the set of fully symmetric distributions with four outcomes in the triangle network that admit a local model, we start by analytically constructing such local models.
A straightforward but potentially limiting way to construct fully symmetric flags is to allow each party only strategies that are symmetric under permutation of the outputs. 
Let us now describe such strategies formally.
\begin{figure}[htb]
    \centering
    \includegraphics{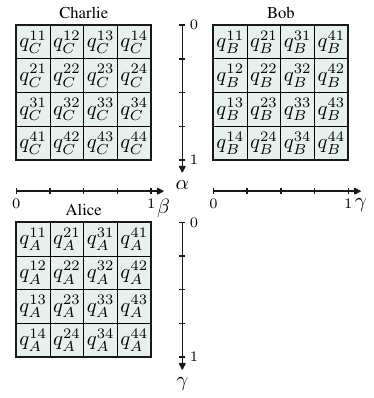}
    \caption{We divide the range of each source into four equal partitions which leads to a $4\times 4$ grid of $16$ strategies per party that we label by $q_A^{jk}$, $q_B^{ki}$ and $q_C^{ij}$, respectively.}
    \label{fig:symmtemplate}
\end{figure}
Let 
\begin{align*}
    q_A^{jk} : \{1,2,3,4\} \times [0,1]^{\times 2} \to [0,1]
\end{align*}
be Alice's ``sub-response-functions'', for $j,k \in \{1,2,3,4\}$.
Let $I_j = [(j-1)/4,j/4[$ for $j \in \{1,2,3,4\}$.
The full response function of Alice is defined as a piecewise function: for all $\beta,\gamma\in[0,1]$, with $j,k$ be such that $\beta \in I_j$ and $\gamma \in I_k$,
\begin{align*}
    p_A(a|\beta,\gamma) = q_A^{jk}(a|4\beta - j + 1,4\gamma - k + 1).
\end{align*}
We now require the following: for all $\sigma \in S_4$ (the group of permutations of $\{1,2,3,4\}$),
\begin{align}
\label{eq:constraint local models}
    q_A^{\sigma(j)\sigma(k)}(\sigma(a)|\beta,\gamma) = q_A^{jk}(a|\beta,\gamma).
\end{align}
We repeat this construction for Bob's and Charlie's response functions, with associated sub-response-functions $q_B^{kl}$ and $q_C^{li}$ satisfying a constraint analogous to \cref{eq:constraint local models}.
A straightforward calculation shows that the output distribution of such a local model is
\begin{align*}
    p(a,b,c) = \frac{1}{4^3} \sum_{i,j,k=1}^4 p^{ijk}(a,b,c),
\end{align*}
where
\begin{align*}
    p^{ijk}(a,b,c) =
    \int_{[0,1]^{\times3}} \dd\alpha\dd\beta\dd\gamma q_A^{jk}(a|\beta,\gamma) q_B^{ki}(b|\gamma,\alpha) q_C^{ij}(c|\alpha,\beta).
\end{align*}
The fact that $p(\sigma(a),\sigma(b),\sigma(c)) = p(a,b,c)$, for all $a,b,c\in\{1,2,3,4\}$ and $\sigma \in S_4$, follows directly from the observation that \cref{eq:constraint local models} implies
\begin{align*}
    p^{ijk}(\sigma(a),\sigma(b),\sigma(c)) = p^{\sigma^{-1}(i)\sigma^{-1}(j)\sigma^{-1}(k)}(a,b,c).
\end{align*}
Thus, this construction yields distributions that are symmetric under permutation of the outputs, but not yet necessarily under permutation of the parties.
We give more details on how to construct such families of response functions in \cref{sec:generating outcome-symmetric local models}.

\subsection{Generating outcome-symmetric local models}
\label{sec:generating outcome-symmetric local models}

In this section, we describe how to generate flags that satisfy the constraints described in \cref{sec:defining outcome-symmetric local models}.
Let $q_A^{11}(a|\beta,\gamma)$ be an arbitrary sub-response-function that satisfies the constraint
\begin{align}
\label{eq:constraint q11}
    q_A^{11}(2|\beta,\gamma) = q_A^{11}(3|\beta,\gamma) = q_A^{11}(4|\beta,\gamma).
\end{align}
Let $\tau_j\in S_4$ be the transposition $1\leftrightarrow j$ (note that $\tau_j^{-1} = \tau_j$ and $\tau_1$ is the identity).
Then, define $q_A^{jj}$ for $j\in\{2,3,4\}$ as follows:
\begin{align}
\label{eq:def qjj}
    q_A^{jj}(a|\beta,\gamma) = q_A^{11}(\tau_j(a)|\beta,\gamma).
\end{align}
Notice that \cref{eq:def qjj} also holds trivially for $j = 1$.
Let $q_A^{12}(a|\beta,\gamma)$ be an arbitrary sub-response-function that satisfies the constraint
\begin{align}
\label{eq:constraint q12}
    q_A^{12}(3|\beta,\gamma) = q_A^{12}(4|\beta,\gamma).
\end{align}
For all $j,k \in \{1,2,3,4\}$, $j \neq k$, define $\pi_{jk} \in S_4$ to be the permutation such that $\pi_{jk}(1) = j$, $\pi_{jk}(2) = k$ (and complete it arbitrarily --- \cref{eq:def qjk} will still be well-defined thanks to \cref{eq:constraint q12}).
Define $q_A^{jk}$ for $j\neq k$, $(j,k) \neq (1,2)$ as follows:
\begin{align}
\label{eq:def qjk}
    q_A^{jk}(a|\beta,\gamma) = q_A^{12}(\pi_{jk}^{-1}(a)|\beta,\gamma).
\end{align}
Notice that \cref{eq:def qjk} also holds trivially for $(j,k) = (1,2)$.
\begin{lemma}
    Any such family $\{q_A^{jk}\}_{j,k=1}^4$ satisfies the constraint of \cref{eq:constraint local models}.
\end{lemma}
\begin{proof}
    \underline{The case of $j = k$.}
    We have to show in this case that
    \begin{align}
        q_A^{jj}(\sigma(a)|\beta,\gamma) = q_A^{\sigma^{-1}(j)\sigma^{-1}(j)}(a|\beta,\gamma).
    \end{align}
    Using \cref{eq:def qjj}, this simplifies to
    \begin{align}
        q_A^{11}(\tau_j \circ \sigma(a)|\beta,\gamma) = q_A^{11}(\tau_{\sigma^{-1}(j)}(a)|\beta,\gamma).
    \end{align}
    Thanks to \cref{eq:constraint q11}, this follows if either 
    \begin{align}
        &(\tau_j\circ\sigma(a) = 1 \textup{ and } \tau_{\sigma^{-1}(j)}(a) = 1) \\
        \textup{or } 
        &(\tau_j\circ\sigma(a) \neq 1 \textup{ and } \tau_{\sigma^{-1}(j)}(a) \neq 1).
    \end{align}
    Let $P$ be the proposition $\tau_j\circ\sigma(a) = 1$ and $Q$ the proposition $\tau_{\sigma^{-1}(j)}(a) = 1$.
    We want to show $(P \land Q) \lor (\lnot P \land \lnot Q)$.
    This is trivially true if $P \Leftrightarrow Q$.
    This is exactly what happens: we have that $P$ is equivalent to $\sigma(a) = j$, while $Q$ is equivalent to $a = \sigma^{-1}(j)$, which is clearly equivalent to $P$.

    \underline{The case of $j \neq k$.}
    We have to show in this case that
    \begin{align}
        q_A^{jk}(\sigma(a)|\beta,\gamma) = q_A^{\sigma^{-1}(j)\sigma^{-1}(k)}(a|\beta,\gamma).
    \end{align}
    Using \cref{eq:def qjk}, this simplifies to
    \begin{align}
        q_A^{12}(\pi^{-1}_{jk}\circ \sigma(a)|\beta,\gamma)
        =
        q_A^{12}(\pi^{-1}_{\sigma^{-1}(j)\sigma^{-1}(k)}(a) |\beta,\gamma).
    \end{align}
    Thanks to \cref{eq:constraint q12}, this follows if either
    \begin{align}
        &(\pi^{-1}_{jk}\circ \sigma(a) = 1 \textup{ and } \pi^{-1}_{\sigma^{-1}(j)\sigma^{-1}(k)}(a) = 1) \\
        \textup{or }
        &(\pi^{-1}_{jk}\circ \sigma(a) = 2 \textup{ and } \pi^{-1}_{\sigma^{-1}(j)\sigma^{-1}(k)}(a) = 2) \\
        \textup{or }
        &(\pi^{-1}_{jk}\circ \sigma(a) \notin \{1,2\} \textup{ and } \pi^{-1}_{\sigma^{-1}(j)\sigma^{-1}(k)}(a) \notin \{1,2\}).
    \end{align}
    Let $P_x$ be the proposition $\pi^{-1}_{jk}\circ\sigma(a) = x$ and $Q_x$ be the proposition $\pi^{-1}_{\sigma^{-1}(j)\sigma^{-1}(k)}(a) = x$.
    We want to show
    \begin{align}
    \label{eq:long logical expression}
        (P_1 \land Q_1) \lor (P_2 \land Q_2) \lor \Big(\lnot(P_1\lor P_2) \land \lnot(Q_1\lor Q_2)\Big).
    \end{align} 
    Notice that $P_1 \Leftrightarrow Q_1$: indeed, $P_1$ is equivalent to $\sigma(a) = \pi_{jk}(1) = j$, i.e., $\sigma(a) = j$, while $Q_1$ is equivalent to $a = \pi_{\sigma^{-1}(j)\sigma^{-1}(k)}(1) = \sigma^{-1}(j)$, i.e., $a = \sigma^{-1}(j)$.
    Similarly, $P_2 \Leftrightarrow Q_2$.
    Thus, \cref{eq:long logical expression} simplifies to
    \begin{align}
        P_1 \lor P_2 \lor \lnot (P_1\lor P_2),
    \end{align}
    which is trivially true.
\end{proof}

\subsection{Maximizing the correlations}
\label{sec:high_SYM_corr}

In this section, we describe flags that satisfy the symmetry constraints of \cref{sec:defining outcome-symmetric local models} and that maximize the correlations $p(A=B=C)$.
Indeed, the combination of high symmetry and strong correlations seems to be one of the essential characteristics implying nonlocality.
We came up with the following parametrized construction, yielding the flags depicted in \cref{fig:symmmaxcorr} and reproduced in \cref{fig:symmmaxcorr2} for easier reference.

From Finner's inequality \cite{Finner1992} (see \cref{eq:general finner}), we know that in order to maximize the volume where $A=B=C$, we should aim at having unicolored ``cubes'' as large as possible. 
Assuming the grid for the symmetrized strategy as in \cref{fig:symmtemplate}, we can at most have such cubes of side length $\sfrac{1}{4}$. 
Thus, it seems that the best strategy would be for each party to output the outcome $k$ in the sub-response-function $q_A^{kk}$, $q_B^{kk}$ and $q_C^{kk}$, respectively.
Next, we need to maximize the correlation in the off-diagonal sub-response-functions $q_A^{jk}$, $q_B^{ki}$ and $q_C^{ij}$, respectively.
We start by considering Alice and Bob and allow them to take a fraction $\nu \in [0,\sfrac{1}{3}]$ (this upper bound on $\nu$ will be explained later) of each interval $I_k = [(k-1)/4,j/4[$ of their $\gamma$ input in which they perfectly correlate by just outputting the color $k$. 
To correlate with Charlie, they output $j$ and $i$ in their remaining parts of the sub-response-functions $q_A^{jk}$ and $q_B^{ki}$, respectively. 
Note that Alice's and Bob's strategies are the same up to a reflection and uniform among all $\alpha \in I_i$ and $\beta \in I_j$.
It remains to define Charlie's sub-response-functions $q_C^{ij}$ for $i \neq j$, which we can describe, thanks to the symmetries in Alice's and Bob's strategies, with just one variable $q \in [0,\sfrac12]$, such that 
$$
    \int_{[0,1]^{\times2}} \dd\alpha\dd\beta q_C^{ij}(i|\alpha,\beta) 
    = \int_{[0,1]^{\times2}} \dd\alpha\dd\beta q_C^{ij}(j|\alpha,\beta) 
    = q,
$$ 
and for all $k \notin \{i,j\}$,
$$
    \int_{[0,1]^{\times2}} \dd\alpha\dd\beta q_C^{ij}(k|\alpha,\beta) 
    = \frac{1}{2}-q.
$$
This yields the strategy given by the flags in \cref{fig:symmmaxcorr2}. 

\begin{figure}[htb]
	\centering
        \includegraphics{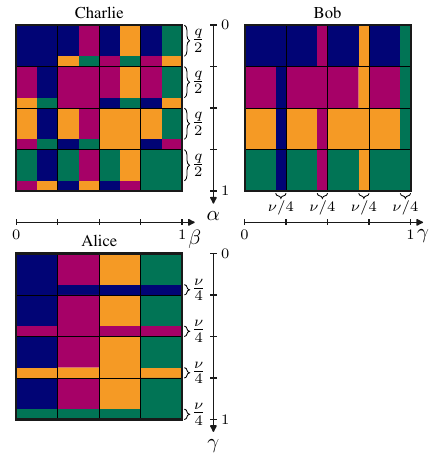}
	\caption{Flags illustrating response functions $a(\beta,\gamma)$, $b(\alpha,\gamma)$ and $c(\alpha,\beta)$ that yield the (conjectured to be) maximal three-party correlation $p(A=B=C)=\sfrac{1}{4}$ as a function of $\nu \in [0,\sfrac{1}{3}]$ and $q = \sfrac{\nu}{(1-\nu)}$.}
	\label{fig:symmmaxcorr2}
\end{figure}

Since the flags satisfy the constraints of \cref{sec:defining outcome-symmetric local models}, it is already guaranteed that the resulting distribution is symmetric under joint permutation of the outcomes.
Note that in \cref{fig:symmmaxcorr2}, the constraint of \cref{eq:constraint q12} is not strictly satisfied: we should have that the \textcolor{fy}{yellow} and \textcolor{fg}{green} components of $q_C^{12}$ are mixed together.
However, the flags that we drew in \cref{fig:symmmaxcorr2} are \emph{equivalent} (as far as the resulting distribution goes) with flags where the \textcolor{fy}{yellow} and \textcolor{fg}{green} regions would be mixed together --- the latter being harder to clearly depict.

Finally, we need to relate $q$ to $\nu$ such that the distribution is symmetric under permutation of the parties. 
The flags as drawn in \cref{fig:symmmaxcorr2} result in the following distribution: for all $k,l,m\in\{1,2,3,4\}$, $k\neq l\neq m\neq k$,
\begin{align*}
	p(k,k,k) &
	= \frac{1}{16}, \\
	p(k,k,l) &= \frac{\nu}{16}, \\
	p(k,l,k) &= p(l,k,k) = (1-\nu)\frac{q}{16}, \\
	p(k,l,m) &= (1-\nu)\frac{1}{16}\Big(\frac{1}{2}-q\Big).
\end{align*}
We must thus require $p(k,k,l) = p(k,l,k)$ to achieve full symmetry, which implies $q = \sfrac{\nu}{(1-\nu)}$.
Noting that $q \in [0,\frac{1}{2}]$, we must thus have $\nu \in [0,\frac{1}{3}]$.
Thus, all distributions resulting from our construction above can be characterized by
\begin{align}
 \label{lambdamax}
	\left(\saaa, \saab, \sabc \right) &= \left( \frac{1}{4},  \frac{9\nu}{4},  \frac{3-9 \nu}{4} \right). 
\end{align}
Building also on the intuition we gained while constructing different local models, this leads us to the conjecture that any local model in the triangle scenario that yields a fully symmetric distribution with four outcomes per party can reach a maximal correlation that is very close to $\saaa = p(A=B=C) = \sfrac{1}{4}$. 
Finding the exact maximum value and proving such an upper bound remains an open question.
A direct implication of this conjecture would be the nonlocality with noise-robustness of the distribution obtained from the Elegant Joint Measurement.
Note that while here we focus on the scenario with four outcomes, we can directly extend our construction to strategies with $N$ outcomes that yield $\saaa^{(N)} = \sfrac{1}{N}$ for $N\geq 3$ (see \cref{app:generalN}) and state a generalized conjecture.

\subsection{More general flags}
\label{app:genflags}

For a more general characterization of the local distributions in the symmetric subspace, we should consider also strategies with smaller correlations or even anti-correlations.
Thus, we allow Alice to choose a fraction $\eta$ out of her fraction $\nu$ in which she previously correlated to Bob, in which she still correlates to Bob, while in the  $(1-\eta)\nu$ remaining fraction she completely anti-correlates to Bob. 
The flag corresponding to her more general strategy is illustrated in \cref{fig:Charliemoregen} (note that Bob's strategy stays the same as before).
In a similar way, Charlie could decide to anti-correlate in his strategies $q_C^{ii}$ (i.e.\ on the diagonal of his flag) by choosing output $i$ only with probability $r \in [0,1]$ and else sample uniformly from the other outputs, as also illustrated in \cref{fig:Charliemoregen}. 
Again, we note that the flag of Charlie depicted in \cref{fig:Charliemoregen} does not strictly satisfy the constraint of \cref{eq:constraint q11}, but is equivalent to one where the \textcolor{fr}{purple}, \textcolor{fy}{yellow} and \textcolor{fg}{green} colors are mixed together in the $q_C^{11}$ sub-response-function.

\begin{figure}[htb]
    \centering
        \includegraphics{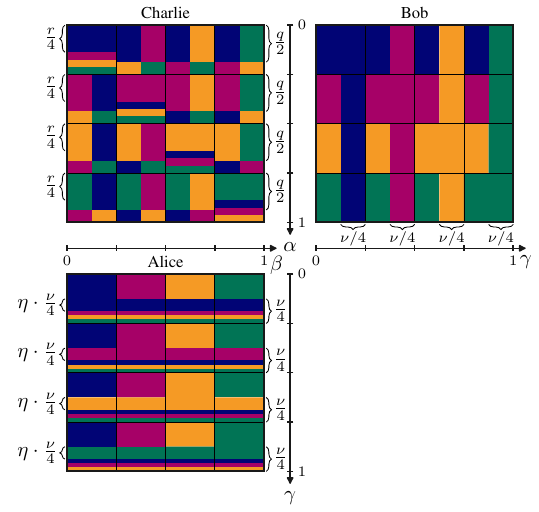}
	\caption{Alice's and Charlie's more general flags, which allow for more anti-correlations. Alice is changing her strategies in the fraction $\nu$, while Charlie is only replacing his strategies $q_C^{ii}$, i.e., those that lie on the diagonal of his grid.}
	\label{fig:Charliemoregen}
\end{figure}

Again, our construction already guarantees that the resulting distribution is symmetric under permutation of the outcomes, but we still need to ensure that it is also symmetric under permutation of the parties. For that we need to relate $q$ to $r$, $\nu$ and $\eta$. Let us determine the following probabilities: for all $k,l,m \in \{1,2,3,4\}$, with $k \neq l \neq m \neq k$,
\begin{align*}
	p(k,k,k) &= \frac{1}{16}\left[(1-\nu)r+\nu\eta\right], \\
	p(k,k,l) &= \frac{1}{16}\left[(1-\nu)\frac{1-r}{3}+\nu\eta\right], \\
	p(k,l,k) &= p(l,k,k) = \frac{1}{16}\left[(1-\nu)q+\nu\frac{1-\eta}{3}\right], \\
	p(k,l,m)&= \frac{1}{16}\left[(1-\nu)(\frac{1}{2}-q)+\nu\frac{1-\eta}{3}\right].
\end{align*}
To achieve symmetry under permutation of the outcomes we require $p(k,k,l) = p(k,l,k)$ which implies 
\begin{equation*}
    q = \frac{1-r}{3}+\frac{\nu}{1-\nu} \frac{4 \eta -1}{3}.
\end{equation*}
Recall that $q$ is subject to the constraint $q \in [0,\sfrac12]$, which implies a constraint on $r,\nu,\eta$.
This yields the distribution in \cref{generallambda}.

\subsection{Even more anti-correlated flags}
\label{app:anticorrflags}

To further explore the set of anti-correlated strategies, a new approach was based on the idea to minimize the probabilities $p(A=B)$, $p(A=C)$ and $p(B=C)$ by completely anti-correlating all off-diagonal elements of the flags.
Using again the construction of \cref{sec:defining outcome-symmetric local models}, this can be done when each party outputs $k$ in their substrategies $q_A^{jk}$, $q_B^{jk}$ and $q_C^{jk}$, respectively, for $j \neq k$. 
Then we only need to determine the ``diagonal'' strategies $q_A^{kk}$, $q_B^{kk}$ and $q_C^{kk}$. 
Uncorrelating Alice and Bob by returning all outputs $i \neq k$ in strategies $q_A^{kk}$ and $q_B^{kk}$ with equal probability independently of $\gamma$ allows to parameterize Charlie's strategy in the way depicted in \cref{fig:anticorrelated}. 
Note that Charlie has only one parameter $r$ to tune the correlation with Alice and Bob, as everything else results from the symmetry constraints.
\begin{figure}[htb]
    \centering
        \includegraphics{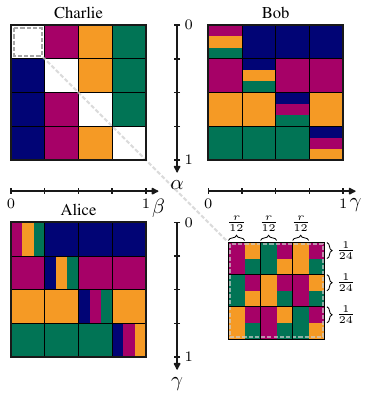}
	\caption{Even more anti-correlated flags that yield the distribution in \cref{eq:bottomrightline}.}
	\label{fig:anticorrelated}
\end{figure}
This leads to the distribution of \cref{eq:bottomrightline}.
To obtain the full ``spike'' at the bottom right of the local distributions shown in \cref{fig:triangleanalytic}, one should furthermore take into the decorrelation currents presented in \cref{app:currents}.

\subsection{Decorrelation currents}
\label{app:currents}

Given a distribution which is guaranteed to be local, but whose local model is unknown, it is occasionally possible to deduce that other distributions are also local.
Such deductions can be obtained by modifying the original unknown local model: the sources and the parties may randomly deviate from their original behavior.
The deviations from the original behavior must be chosen in such a way that the new distribution which is also local can be written entirely in terms of the original output distribution and not in terms of its unknown original local model.
If the deviations from the original model are parametrized by $\varepsilon$, such that for $\varepsilon = 0$, there is no deviation, and for $\varepsilon = 1$, the distribution is ``more mixed'' then the original distribution, then we have what we call a \emph{decorrelation current}, which flows continuously from the original local distribution to other local distributions which are more mixed.

We now give an explicit example of such a decorrelation current.
This is the current that we use to extend the line of local distributions described in \cref{eq:bottomrightline} to the two-dimensional local region labeled $(*)$ in \cref{fig:triangleanalytic}.

\begin{lemma}
    Suppose that the fully symmetric distribution $p^{(0)}$ described by $(\saaa^{(0)},\saab^{(0)},\sabc^{(0)})$ is local.
    Then, for all $\varepsilon \in [0,1]$, for all $l \in [0,1]$, the fully symmetric distribution $p^{(\varepsilon,l)}$ described by
    \begin{align*}
        \saaa^{(\varepsilon,l)} &=
        (1-\varepsilon)^3 \saaa^{(0)} + 3\left(\varepsilon(1-\varepsilon) + \frac14 \varepsilon^3 \right) \frac{1 - l}{4} + \frac{1}{64} \varepsilon^3 \\
        \saab^{(\varepsilon,l)} &= 
        (1-\varepsilon)^3 \saab^{(0)} + 3\left( \varepsilon(1-\varepsilon) + \frac14 \varepsilon^3 \right) \frac{3 - l}{4} + \frac{9}{64} \varepsilon^3 \\
        \sabc^{(\varepsilon,l)} &= 
        (1-\varepsilon)^3 \sabc^{(0)} + 3\left( \varepsilon(1-\varepsilon) + \frac14 \varepsilon^3 \right) \frac{l}{2} + \frac{6}{64} \varepsilon^3
    \end{align*}
    is also local.
\end{lemma}
\begin{proof}
Since the distribution $p^{(0)}$ is local, it admits a local model of the form
$$
    p^{(0)}(a,b,c) = \int_{[0,1]^{\times 3}} \dd\alpha\dd\beta\dd\gamma p_A^{(0)}(a|\beta,\gamma) p_B^{(0)}(b|\gamma,\alpha) p_C^{(0)}(c|\alpha,\beta).
$$
We now define a local model for the distribution $p^{(\varepsilon,l)}$.
We let the three sources distribute tuples of the form 
\begin{align*}
    \textup{$\alpha$ source: } (\alpha, x, b_{BC}, c_{BC}), \\
    \textup{$\beta$ source: } (\beta, y, a_{AC}, c_{AC}), \\
    \textup{$\gamma$ source: } (\gamma, z, a_{AB}, b_{AB}),
\end{align*}
where $\alpha,\beta,\gamma \in [0,1]$, $x,y,z \in \{0,1\}$, and $a_{AB},a_{AC},b_{AB},b_{BC},c_{AC},c_{BC}\in \{1,2,3,4\}$.
This is summarized in \cref{fig:triangle for currents}.
\begin{figure}[ht]
    \centering
    \includegraphics{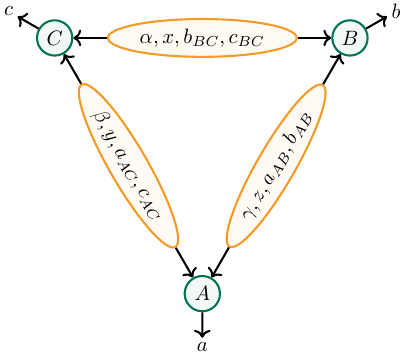}
    \caption{The tuples sent out by the three sources.}
    \label{fig:triangle for currents}
\end{figure}
The source distributions are defined as follows:
\begin{align*}
    p(\alpha,x, b_{BC}, c_{BC}) &= r(x)q(b_{BC}, c_{BC}), \\
    p(\beta,y, a_{AC}, c_{AC}) &= r(y)q(a_{AC},c_{AC}), \\
    p(\gamma,z,a_{AB},b_{AB}) &= r(z)q(a_{AB},b_{AB}),
\end{align*}
i.e., $\alpha,\beta,\gamma$ are still uniformly distributed in $[0,1]$, $x,y,z$ are distributed according to $r$, and the pair of outcomes $a,b$ is distributed as $q(a,b)$.
We let $r(0) = 1 - \varepsilon$ and $r(1) = \varepsilon$.
We define the distribution $q$ as:
\begin{equation}
    \label{eq:def q current}
    q(a,b) = \left\{
        \begin{aligned}
            \frac{1 - l}{4} &\textup{ if } a = b, \\
            \frac{l}{12} &\textup{ if } a\neq b.
        \end{aligned}
    \right.
\end{equation}
We now define the response functions of the parties.
The idea is the following: Alice receives two bits, $y$ and $z$.
If both of those bits are $0$, Alice follows the original strategy looking at $\beta$ and $\gamma$.
If only one of those bits is $1$, Alice outputs the $a$ that is sent along that bit.
If both of those bits are $1$, then Alice chooses uniformly at random whether to output $a_{AB}$ or $a_{AC}$.
Bob and Charlie follow a similar strategy.
This results in the following response functions:
\newcommand{\kd}[2]{\delta_{#1,#2}}
\begin{align*}
    p_A(a|\beta,y, a_{AC}, c_{AC}, \gamma,z, a_{AB}, b_{AB}) &= \kd{y}{0} \kd{z}{0} p_A^{(0)}(a|\beta,\gamma) + \kd{y}{1} \kd{z}{0} \kd{a}{a_{AC}} + \kd{y}{0}\kd{z}{1} \kd{a}{a_{AB}} + \kd{y}{1} \kd{z}{1} \frac{1}{2}(\kd{a}{a_{AB}} + \kd{a}{a_{AC}}), \\
    p_B(b|\gamma,z, a_{AB}, b_{AB}, \alpha, x,  b_{BC}, c_{BC}) &= \kd{x}{0} \kd{z}{0} p_B^{(0)}(b|\gamma,\alpha) + \kd{x}{1} \kd{z}{0} \kd{b}{b_{BC}} + \kd{x}{0}\kd{z}{1} \kd{b}{b_{AB}} + \kd{x}{1} \kd{z}{1} \frac{1}{2}(\kd{b}{b_{AB}} + \kd{b}{b_{BC}}), \\
    p_C(c|\alpha,x, b_{BC}, c_{BC}, \beta,y, a_{AC}, c_{AC}) &= \kd{x}{0} \kd{y}{0} p_C^{(0)}(c|\alpha,\beta) + \kd{x}{1} \kd{y}{0} \kd{c}{c_{BC}} + \kd{x}{0}\kd{y}{1} \kd{c}{c_{AC}} + \kd{x}{1} \kd{y}{1} \frac{1}{2}(\kd{c}{c_{AC}} + \kd{c}{c_{BC}}).
\end{align*}
We used the Kronecker delta $\kd{a}{b}$ which is $1$ if $a=b$ and 0 else.
We can now compute the output distribution of such a local model, averaging over the different cases for the bits $x,y,z$.
One may of course proceed analytically, but looking at individual terms makes the calculation easier.
We summarize this information in \cref{tab:currents}.
\begin{table}[ht]
    \newcommand{\mycell}[1]{\makecell[{{p{5cm}}}]{#1}}
    \centering
    \begin{tabular}{|c|c|c|c|c|}
    \hline
         Case $(x,y,z)$ & Analogous cases & Probability of case & Description & Resulting distribution  \\
    \hline\hline
         $(0, 0, 0)$ & $\emptyset$ & $(1-\varepsilon)^3$ & \mycell{
            The parties simply follow the original local model.
        } & $p^{(0)}(a,b,c)$ \\
        \hline
         $(1, 0, 0)$ & $\begin{aligned}\{&(0,1,0), \\ &(0,0,1) \}\end{aligned}$ & $\varepsilon(1-\varepsilon)^2$ & \mycell{
            Bob and Charlie output $b_{BC}$ and $c_{BC}$,
            distributed according to $q$. Alice follows
            her original strategy but neither Bob nor
            Charlie are looking at $\beta$ and $\gamma$. Her
            output is thus distributed according to the
            marginal $p^{(0)}(a)$.
        } & $q(b,c)p^{(0)}(a)$ \\
    \hline
        $(1,1,0)$ & $\begin{aligned}\{&(1,0,1), \\ &(0,1,1)\}\end{aligned}$ & $\varepsilon^2(1-\varepsilon)$ & \mycell{
            Alice and Bob output $a_{AC}$ and $b_{BC}$,
            respectively.
            Charlie outputs $c_{AC}$ or
            $c_{BC}$ with probability $\sfrac12$.
        }
        & $\frac12\big(q(a,c) q(b) + q(b,c) q(a)\big)$ \\
    \hline
        $(1,1,1)$ & $\emptyset$ & $\varepsilon^3$ & \mycell{
            Alice chooses to output $a_{AB}$ or $a_{AC}$ with
            probability $\sfrac12$, and similarly for Bob and Charlie.
        } & $\begin{gathered}
            \textstyle\frac14 \big(q(a,b)q(c) + q(a,c)q(b) \\
            + \, q(b,c)q(a) + q(a)q(b)q(c) \big)
        \end{gathered}$ \\
    \hline
    \end{tabular}
    \caption{The different cases for the bits $x,y,z$ and the resulting output distributions.}
    \label{tab:currents}
\end{table}
Thus, the output distribution is
\begin{align*}
    p^{(\varepsilon,l)}(a,b,c)
    = (1-\varepsilon)^3 p^{(0)}(a,b,c)
    &+ \varepsilon(1-\varepsilon)^2 \left( q(a,b) p^{(0)}(c) + q(a,c) p^{(0)}(b) + q(b,c) p^{(0)}(a) \right) \\
    &+ \varepsilon^2(1-\varepsilon) \Big( q(a,b) q(c) + q(a,c) q(b) + q(b,c) q(a) \Big) \\
    &+ \frac14 \varepsilon^3 \Big( q(a,b) q(c) + q(a,c) q(b) + q(b,c) q(a) + q(a) q(b) q(c) \Big).
\end{align*}
We now recall the definition of $q$, see \cref{eq:def q current}.
It implies that the marginal $q(a)$ is maximally mixed, i.e., $q(a) = \sfrac14$.
Furthermore, since $p^{(0)}$ was assumed to be fully symmetric, its single party marginals such as $p^{(0)}(a)$ are also maximally mixed.
Thus, we obtain
\begin{align*}
    p^{(\varepsilon,l)}(a,b,c)
    &= (1-\varepsilon)^3 p^{(0)}(a,b,c)
    + \Big(
        \varepsilon(1-\varepsilon)^2 
        + \varepsilon^2(1-\varepsilon) 
        + \frac14 \varepsilon^3
    \Big)
    \frac{q(a,b) + q(a,c) + q(b,c)}{4}
    + \frac{1}{256} \varepsilon^3 \\
    &= (1-\varepsilon)^3 p^{(0)}(a,b,c)
    + \Big(
        \varepsilon(1-\varepsilon)
        + \frac14 \varepsilon^3
    \Big)
    \frac{q(a,b) + q(a,c) + q(b,c)}{4}
    + \frac{1}{256} \varepsilon^3.
\end{align*}
Inserting the definition of $q$, we find that
\begin{align*}
    p^{(\varepsilon,l)}(1,1,1) &= (1-\varepsilon)^3 p^{(0)}(1,1,1) + 3\left(\varepsilon(1-\varepsilon) + \frac14\varepsilon^3\right) \frac{1 - l}{16} + \frac{1}{256} \varepsilon^3, \\
    p^{(\varepsilon,l)}(1,1,2) &= (1-\varepsilon)^3 p^{(0)}(1,1,2) + \left(\varepsilon(1-\varepsilon) + \frac14\varepsilon^3\right) \frac{3-l}{48} + \frac{1}{256} \varepsilon^3, \\
    p^{(\varepsilon,l)}(1,2,3) &= (1-\varepsilon)^3 p^{(0)}(1,2,3) + 3\left(\varepsilon(1-\varepsilon) + \frac14\varepsilon^3\right) \frac{l}{48} + \frac{1}{256} \varepsilon^3.
\end{align*}
Using \cref{eq:def s111}, we obtain the claim.
\end{proof}

\newpage
\begin{figure}[t]
    \centering
    \includegraphics{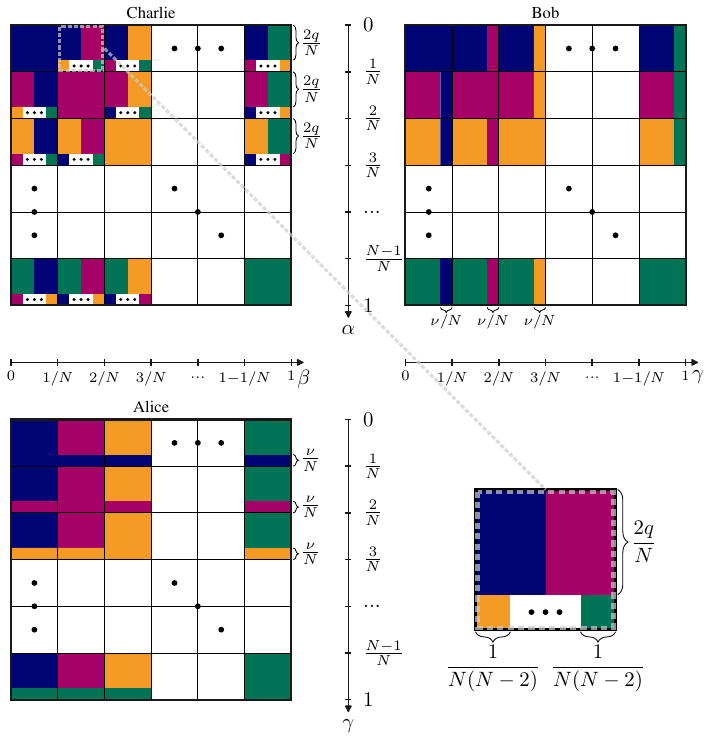}
\caption{The $N$ outcome flags yielding the local distributions of \cref{eq:n outcome distribs}. The color \textcolor{fb}{blue} labels the outcome 0, the color \textcolor{fr}{red} labels the outcome 1, the color \textcolor{fy}{yellow} labels the outcome 2, the outcomes 3 to $N-1$ are left implicit, and the color \textcolor{fg}{green} labels the outcome $N$.}
\label{fig:n outcome flags}
\end{figure}

\section{$N$ outcomes per party}
\label{app:generalN}

We now turn to some investigations of fully symmetric distributions with $N \geq 3$ outcomes.
Analogously to the four outcome case, such fully symmetric distributions can be characterized by the following three numbers:
\begin{align*}
    \saaa = N p(1,1,1), \qquad \saab = 3N(N-1) p(1,1,2), \qquad \sabc = N(N-1)(N-2) p(1,2,3),
\end{align*}
such that $\saaa + \saab + \sabc = 1$.
The corresponding extremal distributions are defined as
\begin{align*}
    p_{111}^{(N)} &= \frac{1}{N} \sum_{k=1}^N [k,k,k], \\
    p_{112}^{(N)} &= \frac{1}{3N(N-1)} \sum_{\substack{k,l=1\\k\neq l}}^N [k,k,l] + [k,l,k] + [l,k,k], \\
    p_{123}^{(N)} &= \frac{1}{N(N-1)(N-2)} \sum_{\substack{k,l,m=1\\k\neq l\neq m\neq k}}^N [k,l,m].
\end{align*}
Notice that in this case, the Finner inequality of \cref{eq:general finner} simplifies to
\begin{align*}
    p(a,b,c) \leq \frac{1}{N^{3/2}},
\end{align*}
since again the marginals of fully symmetric distributions are uniform.
This implies that for all fully symmetric distributions that have either a local or quantum model,
\begin{align}
\label{eq:finner n outcomes}
    \saaa \leq \frac{1}{\sqrt{N}}.
\end{align}

Analogously to the construction for $N=4$, we can generalize our construction of strategies that yield strongly correlated probabilities to $N$ outcomes. 
Similar to the division illustrated in \cref{fig:symmtemplate}, each strategy is divided now in $N$ by $N$ substrategies that are invariant under permutation of the outcomes as in \cref{eq:constraint local models}.
The substrategies need to be adapted to more outcomes, yielding the general construction as depicted in \cref{fig:n outcome flags}, resulting in the following distribution: for all $k,l,m\in\{1,..,N\}$, $k\neq l\neq m\neq k$,
\begin{align*}
	p(k,k,k) &= \frac{1}{N^2}, \\
	p(k,k,l) &= \frac{\nu}{N^2}, \\
	p(k,l,k) &= p(l,k,k) = \frac{q}{N^2}(1-\nu), \\
	p(k,l,m) &= \frac{(1-\nu)(1-2q)}{N^2(N-2)}.
\end{align*}
Again, for symmetry under permutation of the parties we require $p(k,k,l) = p(k,l,k)$, which implies $q = \frac{\nu}{1-\nu}$.
Noting that $q \in [0,\frac{1}{2}]$, we must thus have $\nu \in [0,\frac{1}{3}]$.

Thus, the resulting distributions can be characterized for any $N\geq 3$ by 
\begin{align}
\label{eq:n outcome distribs}
	\left(\saaa, \saab, \sabc \right) &= \left( \frac{1}{N},  \frac{3\nu(N-1)}{N},  \frac{(1-3\nu)(N-1)}{N} \right), \qquad \nu \in \left[0,\frac{1}{3}\right].
\end{align}

We have then also applied LHV-net to characterize the set of local distributions and obtained results shown in \cref{fig:nn cards}. 
Again, these results support the conjecture that classically, the highest $\saaa$ with $N$ outcomes is close to $\sfrac1N$.
\begin{figure}[ht]
    \centering
    \begin{subfigure}{0.49\textwidth}
        \includegraphics{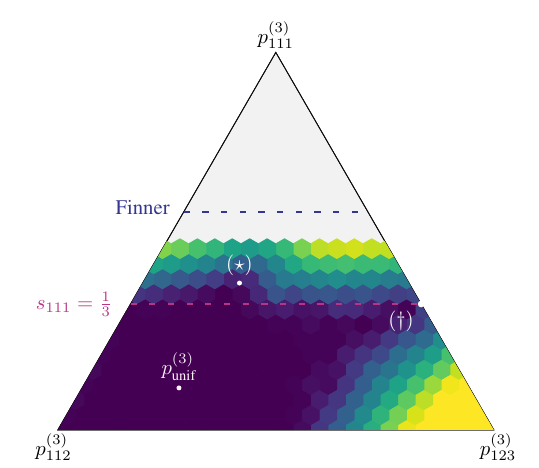}
        \caption{$3$ outcomes per party.}
        \label{fig:nn card 3}
    \end{subfigure}
    \begin{subfigure}{0.49\textwidth}
        \includegraphics{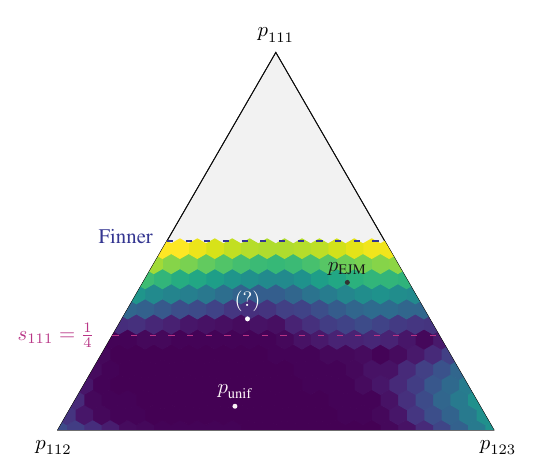}
        \caption{$4$ outcomes per party.}
        \label{fig:nn card 4}
    \end{subfigure}
    \begin{subfigure}{0.49\textwidth}
        \includegraphics{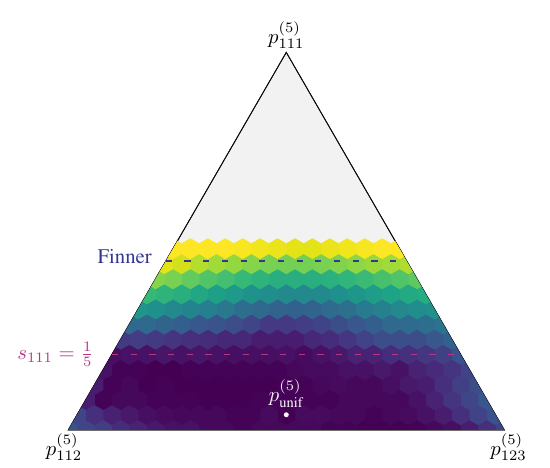}
        \caption{$5$ outcomes per party.}
        \label{fig:nn card 5}
    \end{subfigure}
    \begin{subfigure}{0.49\textwidth}
        \includegraphics{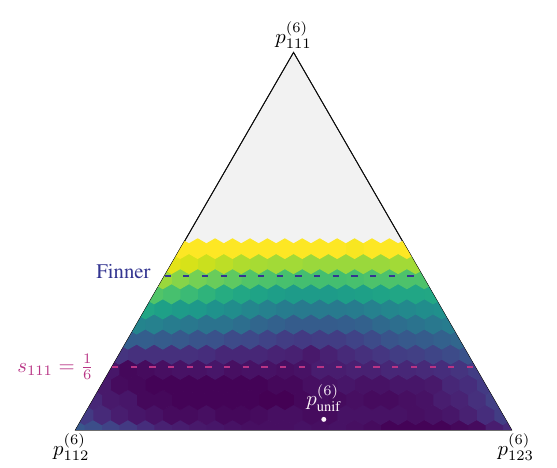}
        \caption{$6$ outcomes per party.}
        \label{fig:nn card 6}
    \end{subfigure}
    \begin{subfigure}{\textwidth}
        \LambdaScope{
            \node at (0,-0.08) {$\phantom{vspacing}$};
            \DrawColorBar
        }
    \end{subfigure}
    \caption{
    Distance of symmetric distributions to the local set according to LHV-Net ($N=3,4,5,6$ for the (a,b,c,d) subplots, respectively, with $\saaa<\sfrac12$ target distributions considered only). For each point, LHV-Net independently minimized the distance to the local set. Color represents Euclidean distance of the numerically found closest (not necessarily symmetric) distribution; values are artificially cut off at $0.1$ in order to provide a collective color scale. Additionally, the Finner inequality and $\sfrac1N$ line is depicted in each map, with the latter being conjectured to be close to the true upper bound.  Additional points: $p_{\text{unif}}$ is the uniform distribution, $(\star)$ marks the counter-example of a classical distribution that has $\saaa>\sfrac1N$ for $N=3$ (see \cref{app:3outcomesExample}); $(\dagger)$ marks the $\saab=0$, $\saaa=\sfrac13$ distribution, which is proven to be the only symmetric distribution with $\saab=0$ for $N=3$ that is local, see \cref{app:3outcomesUniqueExampleIfS112is0}; $(?)$ marks the local distribution with $\saaa>\sfrac14$ for $N=4$, which is not entirely symmetric, but is close (see \cref{app:saaa>0.25}).}
    \label{fig:nn cards}
\end{figure}

\newpage
\section{3 outcomes per party}

\subsection{An example 3-outcome local distribution with $\saaa>\sfrac13$}
\label{app:3outcomesExample}

While the results mentioned earlier in this work seem to indicate that the upper bound for $\saaa$ is close to $\sfrac{1}{N}$ given $N$ outcomes per party, we were able to find a counter example for the case of $N=3$, where we can create classical correlations that yield $\saaa > \sfrac{1}{3}$. The corresponding flags are in \cref{fig:abovebound}, which yield the distribution
\begin{align}
\label{eq:3 outcome example}
	\left(\saaa, \saab, \sabc \right) &= \left( \frac{7}{18},  \frac{7}{18},  \frac{2}{9} \right).
\end{align}
\begin{figure}[htb]
    \centering
    \includegraphics{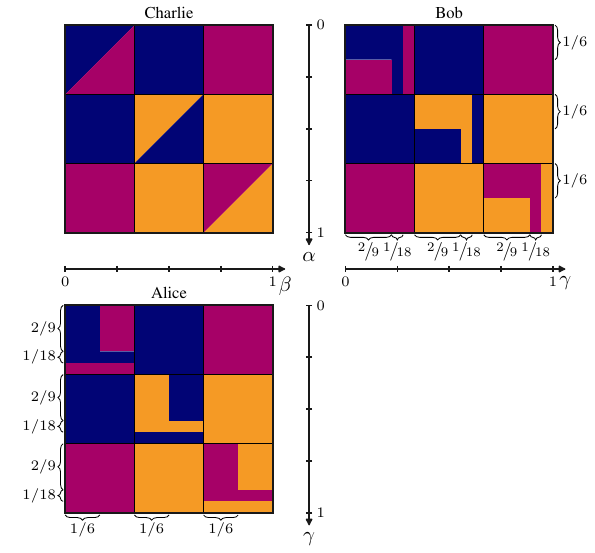}
\caption{The 3 outcome flags yielding the fully symmetric distribution of \cref{eq:3 outcome example}.}
\label{fig:abovebound}
\end{figure}

For more than three outcomes, the generalization of this strategy is not straightforward, since the increased number of combinations for 112-type and 123-type outcomes increases the difficulty of satisfying the symmetry constraints.

\newpage
\subsection{(Almost) unique local strategy for 3 outcomes under the constraint $\saab=0$}
\label{app:3outcomesUniqueExampleIfS112is0}

The Elegant Joint Measurement distribution is particularly fascinating, as it displays a high probability of having $(1,1,1)$-type events, while maintaining a very low probability of having $(1,1,2)$ events. This motivates us to study examples of local distributions where $(1,1,2)$-type events are scarce, or do not appear at all.

We show that for the triangle network with 3 outcomes for each party and under the symmetry constraints (described in \cref{app:generalN}), if one considers distributions with $\saab=0$, then $\saaa = \sfrac13$ is not only an upper bound for $\saaa$, it is also the only possible value for $\saaa$.
Let us call this distribution $p_{\dagger}$, which is portrayed in the symmetric subspace in \cref{fig:nn card 3}. 
Moreover, we show that the local strategy to achieve it is essentially unique, and after reordering the hidden variable values, one flag can be made to have a 3$\times$3 Latin-square structure, where each row and column contains only one of each color, while the other two can be 3$\times$3, 3$\times$2 or 3$\times$1 ``Latin squares''. Technically, a Latin square is an $n \times n$ array, colored with $n$ colors such that each row and column contains exactly one of each color. Hence, when having $n\times m$ grids ($m<n$) we work with generalizations (or cropped versions) of Latin squares.

We will work in the discrete local hidden variable picture with deterministic outcomes labeled by colors (red, green, blue, or $\tR$, $\tG$, $\tB$). We can do this without loss of generality, by using maximally 24 symbols for each local hidden variable~\cite{rosset_2017_universal}. With a slight abuse of notation, we will label the set of all symbols of source $\alpha$ as $\alpha$, and similarly for $\beta,\gamma$.

\subsubsection{Qualitative response functions}
We will first show that the response functions must take the Latin square structure. The steps are as follow.

\noindent\textbf{Step 1. Choosing an RRR.} First, start with a triple of symbols $(\alpha_{R0},\beta_{R0},\gamma_{R0})$, where $\alpha_{R0}\in \alpha,\beta_{R0}\in\beta$, and $\gamma_{R0}\in\gamma$, such that the outcome is $\tR\tR\tR$ (formally, $a(\beta_{R0},\gamma_{R0})=\tR$, $b(\gamma_{R0},\alpha_{R0})=\tR$, and $c(\alpha_{R0},\beta_{R0})=\tR$) . These $\tR$ responses are portrayed in \cref{fig:15a}.

\noindent\textbf{Step 2. Rearranging $\gamma$ symbols.} Next, notice that in Alice's flag, for any $\gamma' \neq \gamma_{R0}$, the color of Alice's response $a(\beta_{R0},\gamma')$ uniquely determines Bob's response $b(\gamma',\alpha_{R0})$, according to the following table.
\begin{table}[h!]
\centering
\begin{tabular}{c c c}
\hline
$a(\beta_{R0},\gamma')$ & & $b(\gamma',\alpha_{R0})$ \\ \hline
$\tR$ & $\implies$ & $\tR$\\
$\tG$ & $\implies$ & $\tB$\\
$\tB$ & $\implies$ & $\tG$
\end{tabular}
\end{table}

\noindent This is true due to the strict $\saab=0$ condition, and since for all these cases Charlie's response is $\tR$, i.e.\ $c(\alpha_{R0},\beta_{R0}) = \tR$. In some sense, this can be seen as a ``color inversion around $\tR$'', induced by Charlie's $\tR$, which leaves $\tR$ invariant, but flips $\tG$ and $\tB$.

So let us now rearrange the symbols of $\gamma$, such that on Alice's flag, in the column defined by $\beta_{R0}$, we first have $\tR$ responses below Alice's first $\tR$ rectangle, and then $\tG$ responses, and finally all the $\tB$ responses, as shown in \cref{fig:15b}. Due to the table above, this immediately implies that Bob's first row will be ordered as $\tR,\tB, \tG$, as shown in \cref{fig:15b}.

\noindent\textbf{Step 3. Rearranging $\beta$ symbols.} We can apply the same procedure to the $\beta$ axis, grouping Alice's responses in her first $\gamma_{R0}$ row in the order $\tR, \tG, \tB$. This implies, due to Bob's $\tR$ in his $\alpha_{R0}$ row, that Charlie's first $\alpha_{R0}$ row must be in the order $\tR,\tB,\tG$. This is illustrated in \cref{fig:15c}.

\noindent\textbf{Step 4. Filling the rest of Alice's flag.} The rest of Alice's flag is uniquely determined by what we have already established on Bob and Charlie's flags, and the $\saab=0$ condition, i.e. by the table used in Step 2, and similar versions for $\tG$ and $\tB$. The filling is shown in \cref{fig:15d}.

\noindent\textbf{Step 5. Bob and Charlie's rows (rearranging $\alpha$ symbols)}. Before moving on, note that even with Bob and Charlie having only a single row, we have all (1,1,1)-type and (1,2,3)-type outcome events appearing, namely ($\tR\tR\tR$, $\tR\tG\tB$, $\tR\tB\tG$, $\tG\tG\tG$, $\tG\tR\tB$, $\tG\tB\tR$, $\tB\tB\tB$, $\tB\tR\tG$, $\tB\tG\tR$). In fact, as we later show, if one stretches the first row of Bob and Charlie to cover their whole flags ($\alpha_{R0}=1)$, one can already obtain $p_\dagger$.

Again, we may rearrange $\alpha$ such that Bob has the order of $\tR,\tG,\tB$ in his first column (as shown in \cref{fig:15e} and \cref{fig:15f}). Any $\tR$ below Bob's first $\tR$ will generate the same row as before on both Bob and Charlie's flag, whereas $\tG$ below this first $\tR$ will force Charlie to have a $\tB,\tG,\tR$ row, and Bob to have a $\tG,\tR,\tB$ row (\cref{fig:15e}. 

In a similar manner, if there is $\tB$ anywhere in Bob's first column, then Charlie's row becomes $\tG,\tR,\tB$, and Bob's is forced to be $\tB,\tG,\tR$, as shown in \cref{fig:15f}.

\begin{figure}[t!]
\renewcommand{\FlagSize}{55pt}
\renewcommand{\FlagSeparation}{25pt}
\begin{subfigure}{0.33\textwidth}
    \begin{tikzpicture}
        \FillAliceFlag{0}{0}{0.2}{0.3333333}{fr};
        \FillBobFlag{0}{0}{0.3}{0.2}{fr};
        \FillCharlieFlag{0}{0}{0.3}{0.3333333}{fr};
        \DrawFlagCanvas;
    \end{tikzpicture}%
    \caption{}
    \label{fig:15a}
\end{subfigure}
\begin{subfigure}{0.33\textwidth}
\begin{tikzpicture}

    \FillAliceFlag{0}{0}{0.3}{0.3333333}{fr};
    \FillAliceFlag{1.5}{0}{0.2}{0.3333333}{fg};
    \FillAliceFlag{1}{0}{0.5}{0.3333333}{fb};
    
    \FillBobFlag{0}{0}{0.3333333}{0.3}{fr};
    \FillBobFlag{0}{1.5}{0.3333333}{0.2}{fb};
    \FillBobFlag{0}{1}{0.3333333}{0.5}{fg};
    
    \FillCharlieFlag{0}{0}{0.3}{0.3333333}{fr};
    
    \DrawFlagCanvas;

\end{tikzpicture}%
    \caption{}
    \label{fig:15b}
\end{subfigure}
\begin{subfigure}{0.33\textwidth}
\begin{tikzpicture}

    \FillAliceFlag{0}{0}{0.3}{0.3333333}{fr};
    \FillAliceFlag{1.5}{0}{0.2}{0.3333333}{fg};
    \FillAliceFlag{1}{0}{0.5}{0.3333333}{fb};
    
    \FillAliceFlag{0}{3.3333333}{0.2}{0.1}{fr};
    \FillAliceFlag{0}{1.6333333}{0.2}{0.266666666}{fg};
    \FillAliceFlag{0}{2.3333333}{0.2}{0.3}{fb};
    
    \FillBobFlag{0}{0}{0.3333333}{0.3}{fr};
    \FillBobFlag{0}{1.5}{0.3333333}{0.2}{fb};
    \FillBobFlag{0}{1}{0.3333333}{0.5}{fg};
    
    \FillCharlieFlag{0}{0}{0.3333333}{0.4333333}{fr};
    \FillCharlieFlag{0}{1.6333333}{0.3333333}{0.266666666}{fb};
    \FillCharlieFlag{0}{2.3333333}{0.3333333}{0.3}{fg};
    
    \DrawFlagCanvas;

\end{tikzpicture}%
    \caption{}
    \label{fig:15c}
\end{subfigure}
\begin{subfigure}{0.33\textwidth}
\begin{tikzpicture}

    \FillAliceFlag{0}{0}{0.3}{0.4333333}{fr};
    \FillAliceFlag{1.5}{0}{0.2}{0.4333333}{fg};
    \FillAliceFlag{1}{0}{0.5}{0.4333333}{fb};
    
    \FillAliceFlag{1.5}{1.633333}{0.2}{0.26666666}{fb};
    \FillAliceFlag{1}{1.633333}{0.5}{0.26666666}{fr};

    \FillAliceFlag{1.5}{2.333333}{0.2}{0.3}{fr};
    \FillAliceFlag{1}{2.333333}{0.5}{0.3}{fg};
    
    \FillAliceFlag{0}{1.6333333}{0.3}{0.266666666}{fg};
    \FillAliceFlag{0}{2.3333333}{0.3}{0.3}{fb};
    
    \FillBobFlag{0}{0}{0.3333333}{0.3}{fr};
    \FillBobFlag{0}{1.5}{0.3333333}{0.2}{fb};
    \FillBobFlag{0}{1}{0.3333333}{0.5}{fg};
    
    \FillCharlieFlag{0}{0}{0.3333333}{0.4333333}{fr};
    \FillCharlieFlag{0}{1.6333333}{0.3333333}{0.266666666}{fb};
    \FillCharlieFlag{0}{2.3333333}{0.3333333}{0.3}{fg};
    \DrawFlagCanvas;

\end{tikzpicture}%
    \caption{}
    \label{fig:15d}
\end{subfigure}
\begin{subfigure}{0.33\textwidth}
\begin{tikzpicture}
    \FillAliceFlag{0}{0}{0.3}{0.4333333}{fr};
    \FillAliceFlag{1.5}{0}{0.2}{0.4333333}{fg};
    \FillAliceFlag{1}{0}{0.5}{0.4333333}{fb};
    
    \FillAliceFlag{1.5}{1.633333}{0.2}{0.26666666}{fb};
    \FillAliceFlag{1}{1.633333}{0.5}{0.26666666}{fr};

    \FillAliceFlag{1.5}{2.333333}{0.2}{0.3}{fr};
    \FillAliceFlag{1}{2.333333}{0.5}{0.3}{fg};
    
    \FillAliceFlag{0}{1.6333333}{0.3}{0.266666666}{fg};
    \FillAliceFlag{0}{2.3333333}{0.3}{0.3}{fb};
    
    \FillBobFlag{0}{0}{0.3333333}{0.3}{fr};
    \FillBobFlag{0}{1.5}{0.3333333}{0.2}{fb};
    \FillBobFlag{0}{1}{0.3333333}{0.5}{fg};

    \FillBobFlag{2}{0}{0.16666666}{0.3}{fg};
    \FillBobFlag{2}{1.5}{0.16666666}{0.2}{fr};
    \FillBobFlag{2}{1}{0.16666666}{0.5}{fb};
    
    \FillCharlieFlag{0}{0}{0.3333333}{0.4333333}{fr};
    \FillCharlieFlag{0}{1.6333333}{0.3333333}{0.266666666}{fb};
    \FillCharlieFlag{0}{2.3333333}{0.3333333}{0.3}{fg};

    \FillCharlieFlag{2}{0}{0.16666666}{0.4333333}{fb};
    \FillCharlieFlag{2}{1.6333333}{0.16666666}{0.266666666}{fg};
    \FillCharlieFlag{2}{2.3333333}{0.16666666}{0.3}{fr};
    
    \DrawFlagCanvas;
\end{tikzpicture}%
    \caption{}
    \label{fig:15e}
\end{subfigure}
\begin{subfigure}{0.33\textwidth}
\begin{tikzpicture}

    \FillAliceFlag{0}{0}{0.3}{0.4333333}{fr};
    \FillAliceFlag{1.5}{0}{0.2}{0.4333333}{fg};
    \FillAliceFlag{1}{0}{0.5}{0.4333333}{fb};
    
    \FillAliceFlag{1.5}{1.633333}{0.2}{0.26666666}{fb};
    \FillAliceFlag{1}{1.633333}{0.5}{0.26666666}{fr};

    \FillAliceFlag{1.5}{2.333333}{0.2}{0.3}{fr};
    \FillAliceFlag{1}{2.333333}{0.5}{0.3}{fg};
    
    \FillAliceFlag{0}{1.6333333}{0.3}{0.266666666}{fg};
    \FillAliceFlag{0}{2.3333333}{0.3}{0.3}{fb};
    
    \FillBobFlag{0}{0}{0.3333333}{0.3}{fr};
    \FillBobFlag{0}{1.5}{0.3333333}{0.2}{fb};
    \FillBobFlag{0}{1}{0.3333333}{0.5}{fg};

    \FillBobFlag{2}{0}{0.16666666}{0.3}{fg};
    \FillBobFlag{2}{1.5}{0.16666666}{0.2}{fr};
    \FillBobFlag{2}{1}{0.16666666}{0.5}{fb};

    \FillBobFlag{1}{0}{0.5}{0.3}{fb};
    \FillBobFlag{1}{1.5}{0.5}{0.2}{fg};
    \FillBobFlag{1}{1}{0.5}{0.5}{fr};
    
    \FillCharlieFlag{0}{0}{0.3333333}{0.4333333}{fr};
    \FillCharlieFlag{0}{1.6333333}{0.3333333}{0.266666666}{fb};
    \FillCharlieFlag{0}{2.3333333}{0.3333333}{0.3}{fg};

    \FillCharlieFlag{2}{0}{0.16666666}{0.4333333}{fb};
    \FillCharlieFlag{2}{1.6333333}{0.16666666}{0.266666666}{fg};
    \FillCharlieFlag{2}{2.3333333}{0.16666666}{0.3}{fr};

    \FillCharlieFlag{1}{0}{0.5}{0.4333333}{fg};
    \FillCharlieFlag{1}{1.6333333}{0.5}{0.266666666}{fr};
    \FillCharlieFlag{1}{2.3333333}{0.5}{0.3}{fb};
    
    \DrawFlagCanvas;

\end{tikzpicture}
    \caption{}
    \label{fig:15f}
\end{subfigure}

\caption{(a) Step 1: Move an $\tR\tR\tR$ event into one corner.
(b) Step 2: Rearrange $\gamma$, such that Alice has the order of $\tR,\tG,\tB$ in her first column. Due to the $\saab=0$ condition, this implies what Bob's first row must look like.
(c) Step 3: Rearrange $\beta$, such that Alice has the order of $\tR,\tG,\tB$ in her first row. This implies what Charlie's first row must look like.
(d) Step 4: The rest of Alice's flag is uniquely determined by the $\saab=0$ condition.
(e) Step 5: Rearrange $\alpha$ s.t. the $\tG$ events come below the $\tR$ events in Bob's first column. A single $\tG$ event uniquely determines both Bob and Charlie's second row. (Note that in the figure, we have not drawn any $\tR$ event below the original one, however, this would just imply a wider first row for Bob and Charlie, with the same coloring.)
(f) In a similar manner, a $\tB$ event in Bob's flag below Bob's first (top left) $\tR$ rectangle determines both Bob and Charlie's final row. At this point we have exhausted all possibilities for coloring the flags.}
\label{fig:3outcome_no112_flags_1}
\end{figure}
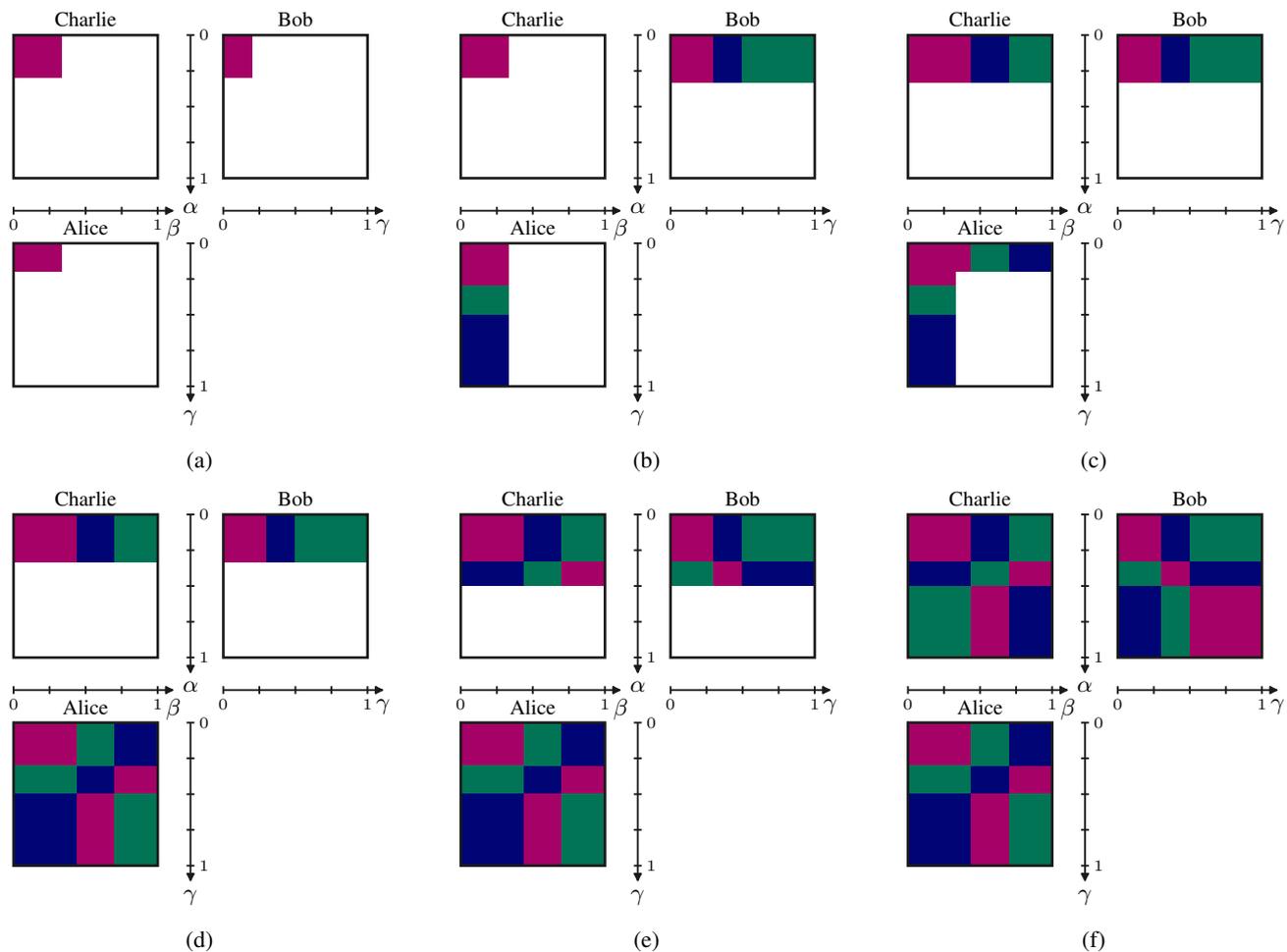

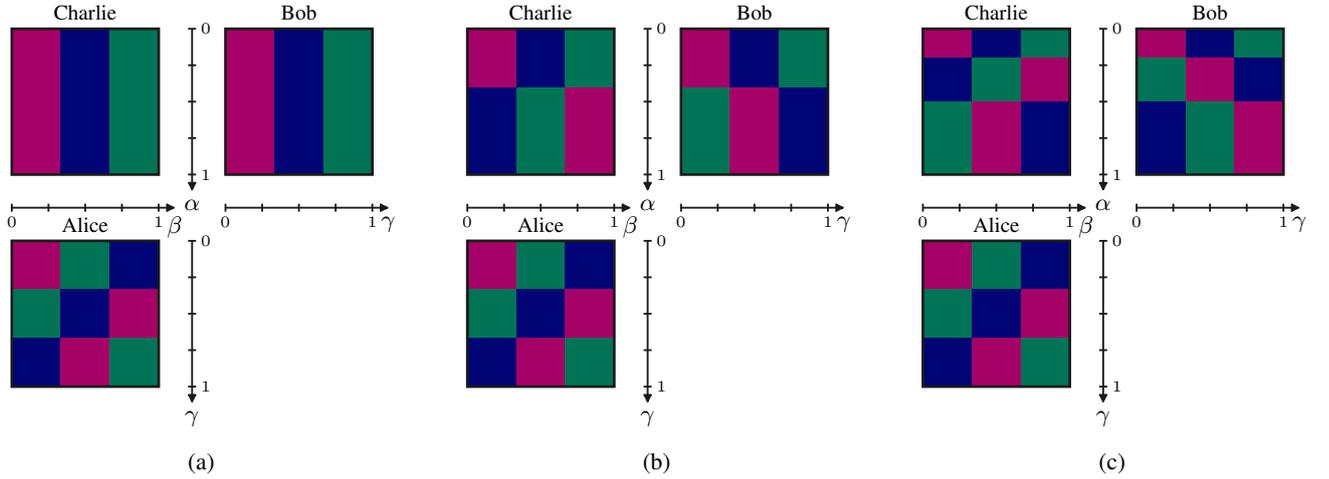
\begin{figure}[h!]
\renewcommand{\FlagSize}{55pt}
\renewcommand{\FlagSeparation}{25pt}
\begin{subfigure}{0.33\textwidth}
\begin{tikzpicture}

    \FillAliceFlag{0}{0}{0.3333333}{0.3333333}{fr};
    \FillAliceFlag{0}{1}{0.3333333}{0.3333333}{fg};
    \FillAliceFlag{0}{2}{0.3333333}{0.3333333}{fb};
    \FillAliceFlag{1}{0}{0.3333333}{0.3333333}{fg};
    \FillAliceFlag{1}{1}{0.3333333}{0.3333333}{fb};
    \FillAliceFlag{1}{2}{0.3333333}{0.3333333}{fr};
    \FillAliceFlag{2}{0}{0.3333333}{0.3333333}{fb};
    \FillAliceFlag{2}{1}{0.3333333}{0.3333333}{fr};
    \FillAliceFlag{2}{2}{0.3333333}{0.3333333}{fg};

    \FillBobFlag{0}{0}{1}{0.3333333}{fr};
    \FillBobFlag{0}{1}{1}{0.3333333}{fb};
    \FillBobFlag{0}{2}{1}{0.3333333}{fg};
    
    \FillCharlieFlag{0}{0}{1}{0.3333333}{fr};
    \FillCharlieFlag{0}{1}{1}{0.3333333}{fb};
    \FillCharlieFlag{0}{2}{1}{0.3333333}{fg};
    
    \DrawFlagCanvas;

\end{tikzpicture}%
\caption{}\label{fig:16a}
\end{subfigure}
\begin{subfigure}{0.33\textwidth}
\begin{tikzpicture}

    \FillAliceFlag{0}{0}{0.3333333}{0.3333333}{fr};
    \FillAliceFlag{0}{1}{0.3333333}{0.3333333}{fg};
    \FillAliceFlag{0}{2}{0.3333333}{0.3333333}{fb};
    \FillAliceFlag{1}{0}{0.3333333}{0.3333333}{fg};
    \FillAliceFlag{1}{1}{0.3333333}{0.3333333}{fb};
    \FillAliceFlag{1}{2}{0.3333333}{0.3333333}{fr};
    \FillAliceFlag{2}{0}{0.3333333}{0.3333333}{fb};
    \FillAliceFlag{2}{1}{0.3333333}{0.3333333}{fr};
    \FillAliceFlag{2}{2}{0.3333333}{0.3333333}{fg};
    
    \FillBobFlag{0}{0}{0.4}{0.3333333}{fr};
    \FillBobFlag{0}{1}{0.4}{0.3333333}{fb};
    \FillBobFlag{0}{2}{0.4}{0.3333333}{fg};
    
    \FillBobFlag{0.666666}{0}{0.6}{0.3333333}{fg};
    \FillBobFlag{0.666666}{1}{0.6}{0.3333333}{fr};
    \FillBobFlag{0.666666}{2}{0.6}{0.3333333}{fb};    
    
    \FillCharlieFlag{0}{0}{0.4}{0.3333333}{fr};
    \FillCharlieFlag{0}{1}{0.4}{0.3333333}{fb};
    \FillCharlieFlag{0}{2}{0.4}{0.3333333}{fg};
    
    \FillCharlieFlag{0.666666}{0}{0.6}{0.3333333}{fb};
    \FillCharlieFlag{0.666666}{1}{0.6}{0.3333333}{fg};
    \FillCharlieFlag{0.666666}{2}{0.6}{0.3333333}{fr};
    
    \DrawFlagCanvas;

\end{tikzpicture}%
\caption{}\label{fig:16b}
\end{subfigure}
\begin{subfigure}{0.33\textwidth}
\begin{tikzpicture}

    \FillAliceFlag{0}{0}{0.3333333}{0.3333333}{fr};
    \FillAliceFlag{0}{1}{0.3333333}{0.3333333}{fg};
    \FillAliceFlag{0}{2}{0.3333333}{0.3333333}{fb};
    \FillAliceFlag{1}{0}{0.3333333}{0.3333333}{fg};
    \FillAliceFlag{1}{1}{0.3333333}{0.3333333}{fb};
    \FillAliceFlag{1}{2}{0.3333333}{0.3333333}{fr};
    \FillAliceFlag{2}{0}{0.3333333}{0.3333333}{fb};
    \FillAliceFlag{2}{1}{0.3333333}{0.3333333}{fr};
    \FillAliceFlag{2}{2}{0.3333333}{0.3333333}{fg};
    
    \FillBobFlag{0}{0}{0.2}{0.3333333}{fr};
    \FillBobFlag{0}{1}{0.2}{0.3333333}{fb};
    \FillBobFlag{0}{2}{0.2}{0.3333333}{fg};
    
    \FillBobFlag{0.666666}{0}{0.3}{0.3333333}{fg};
    \FillBobFlag{0.666666}{1}{0.3}{0.3333333}{fr};
    \FillBobFlag{0.666666}{2}{0.3}{0.3333333}{fb};
    
    \FillBobFlag{1}{0}{0.5}{0.3333333}{fb};
    \FillBobFlag{1}{1}{0.5}{0.3333333}{fg};
    \FillBobFlag{1}{2}{0.5}{0.3333333}{fr};    
    
    \FillCharlieFlag{0}{0}{0.2}{0.3333333}{fr};
    \FillCharlieFlag{0}{1}{0.2}{0.3333333}{fb};
    \FillCharlieFlag{0}{2}{0.2}{0.3333333}{fg};
    
    \FillCharlieFlag{0.666666}{0}{0.3}{0.3333333}{fb};
    \FillCharlieFlag{0.666666}{1}{0.3}{0.3333333}{fg};
    \FillCharlieFlag{0.666666}{2}{0.3}{0.3333333}{fr};

    \FillCharlieFlag{1}{0}{0.5}{0.3333333}{fg};
    \FillCharlieFlag{1}{1}{0.5}{0.3333333}{fr};
    \FillCharlieFlag{1}{2}{0.5}{0.3333333}{fb};
    \DrawFlagCanvas;

\end{tikzpicture}%
\caption{}\label{fig:16c}
\end{subfigure}
\caption{All possible flags that reproduce $p_{\dagger}$ for 3 outcomes per party, (up to permutations of colors, parties and local hidden variable values). Note that any $\{p(\alpha_i)\}_{i=1}^3$ values are valid as long as they sum to one, as as illustrated in the examples (a-c), where (a) $p(\alpha_2)=p(\alpha_3)=0$, (b) $p(\alpha_3=0)$, and (c) $p(\alpha_i) > 0, \forall i$. Recall that $p(\alpha_1)$ is the height of the $\tR,\tB,\tG$ row of Charlie, $p(\alpha_2)$ is the height of the $\tB,\tG,\tR$ row of Charlie, and $p(\alpha_3)$ is the height of the $\tG,\tR,\tB$ row of Charlie.}
\label{fig:3outcome_no112_flags_2}
\end{figure}

\subsubsection{Cardinalities}

Before moving on to fixing the length of the rectangle sides within the flags, let us stop for a note on the cardinalities of the local hidden variables. For 3 outcomes per party, there are three 111-type events (e.g., $\tR\tR\tR$), and six 123-type events (e.g., $\tR\tG\tB$). That means that in the whole cube (as parametrized by $\alpha,\beta,\gamma$), there must be at least 9 distinct types of sub-cuboids ($\tR\tR\tR$, $\tR\tG\tB$, $\tR\tB\tG$, $\tG\tG\tG$, $\tG\tR\tB$, $\tG\tB\tR$, $\tB\tB\tB$, $\tB\tR\tG$, $\tB\tG\tR$). In order to generate these cuboids, the product of the three local hidden variable cardinalities must be at least 9, e.g., $|\alpha|=1, |\beta|=3, |\gamma|=3$ is already sufficient (if it can satisfy all the desired constraints, which it can as  we have seen in \cref{fig:16a}). Smaller cardinalities, e.g., $(1,1,9)$ or $(1,2,5)$ are insufficient because each party (Alice, Bob and Charlie) must output at least 3 distinct colors. 
The above reasoning implies that for any distribution where $\saaa >0 $ and $\sabc>0$, at least one party's flag is at least a 3x3 grid (with no two columns or two rows being the same in the flag). Hence the construction of Alice's flag in the previous sections was not merely a demonstration, but it was necessary that she had these many different colors in her first row and column.

Also note that, as we have shown, due to the limited possibilities when enforcing $\saab=0$, (3,3,3) is effectively the largest LHV cardinality triple. Anything larger can be rearranged to take this 3 by 3 Latin square form.

\subsubsection{Lengths of the rectangle sides}

We start from an effectively $3\times 3 \times 3$ grid, with the elements of $\{\alpha_i\}_{i=1}^3, \{\beta_i\}_{i=1}^3$ and  $\{\gamma_i\}_{i=1}^3$ denoting the symbols of the respective $\alpha,\beta$ or $\gamma$ random variables.

Let us now consider the lengths of the rectangle sides, i.e. the probabilities of $\alpha_i, \beta_j, \gamma_k$ appearing. 
First, let Bob and Charlie only have stripes, i.e., $p(\alpha_2) = p(\alpha_3) = 0$, as in \cref{fig:16a}. Then from the fact that the marginals must all be $\sfrac13$ on Bob and Charlie's flags, we immediately get that $\forall i: \,p(\beta_i)=p(\gamma_i)=\sfrac13$. 
Extending this strategy to the case that $p(\alpha_2)>0$ and $p(\alpha_3) \geq 0$, we see that next to these $p(\beta_i), p(\gamma_i)$ values, one can have any $p(\alpha_1), p(\alpha_2),p(\alpha_3) \geq 0$ such that $\sum_i p(\alpha_i) =1$. 
However, can $p(\beta_i)$ and $p(\gamma_i)$ take on different values than $\sfrac13$? No, as this is already guaranteed by the symmetry constraints for the marginals on Alice's flag. 
If any of the rectangles is larger, then the other rectangles will be smaller and will not be able to satisfy the constraint that $p(A=\tR)=p(A=\tG)=p(A=\tB)=\sfrac13$. 
An optimization in the Mathematica software confirms this intuition. 
More precisely, setting the constraints that the marginals must be $1/3$ for a given flag and positivity and normalization of the hidden variables' probabilities, we get that at least one of its sources (local hidden variables) must have equal values (e.g. $\beta_1=\beta_2=\beta_3)$. This is true for each of the three flags, and since there are three of them, at least two hidden variables must have this property. As a consequence, at least one of the flags must be composed of a 3$\times$3 grid of equal sided squares.

In summary, any local distribution with 3 outcomes, for which $\saab=0$ and for which the symmetry constraints are satisfied (as described in \cref{app:generalN}), must have $\saaa=1/3$. 
The only type of strategy is the Latin square strategy, with $p(\beta_i) = p(\gamma_i) = 1/3$ for $i=1,2,3$, and for any $p(\alpha_1), p(\alpha_2),p(\alpha_3) \geq 0$ s.t. $\sum_i p(\alpha_i) =1$. These strategies are depicted in \cref{fig:3outcome_no112_flags_2}.

Hopefully such an explicit strategy can serve as a basis point for understanding the nonlocality of distributions such as the Elegant Joint Measurement distribution, which has a large $\saaa$ value and a small (though non-zero) $\saab$ value. 
In this aspect, note that the generalization of these arguments to, e.g., 4 outputs per party is not so straightforward, as having a $\tR$ and $\tG$ on two flags implies that the third can be $\tY$ or $\tB$, i.e., it is not fixed uniquely. Similar strategies as these, however, can be constructed, though they are not as unique, as can be seen in the main text in \cref{fig:symmmaxcorr}.

\newpage
\section{Finding inequality parameters with LHV-Net, and perhaps a $\saaa>\sfrac14$ local distribution?}
\label{app:inequality_params}

LHV-Net, the neural network which parametrizes local models obeying the triangle structure, can be used to maximize the left hand side of the inequality of \cref{eq:ineq:basic}, namely of 
\begin{align}
f_w(p) \leq f_w(\pejm) - \delta_w,
\end{align}
where $f_w(p) = w \cdot \saaa(p) - (1-w) \Delta_{l}(p)$ is a function that is large for large $\saaa$ and for symmetric distributions ($\Delta_{l}$ quantifies asymmetry). A strictly positive $\delta_w$ indicates that the EJM distribution outperforms local models for this inequality. 
For several different values of $w$ we see whether LHV-Net can outperform the value of the EJM distribution. 
We plot $\max_{\{p \text{ by LHV-Net}\}} f_w(p) - f_w(\pejm)$ in \cref{fig:LHVnet:inequality_test}. 
We can see that for a range of $w$ values LHV-Net can \textit{not} go over the value 0, hence for these $w$ values we can conjecture inequalities.
The corresponding gaps, $\delta_w$, can be read off from the vertical distance between the blue line and LHV-Net's result for these $w$'s. 
There is no unique best inequality in general, however, in the maintext we chose those $w$ values where $\delta_w$ was largest, namely $w^*=0.678$, $\delta_{w^*}=0.069$ for $l=1$, and $w^*=0.16$, $\delta_{w^*} = 0.012$ for $l=2$,

\subsection{Absolute value penalty inequalities ($l=1$)}
When examining the plot for $l=1$, we find that there seem to be two natural regimes: 1) where the asymmetry is strongly punished (small $w$ values): in this regime fully symmetric local models with high $\saaa$ seem to perform best; and 2) where asymmetry is not punished strongly (large $w$ values), for which the all-111 strategy performs best, where all parties always output 1 (or blue). 
The performance of these two extremal strategies are also depicted in \cref{fig:LHVnet:inequality_test}.
Notice that LHV-Net seems to find a distribution with $\saaa>\sfrac14$. 
This can be seen from the slope of the LHV-Net's results as a function of $w$. 
From the distributions found by LHV-Net we extract that $\saaa\approx 0.289$ could maybe be the maximal $\saaa$ within the symmetric subspace. 

A strategy found by the neural network for $\saaa\approx 0.294$ ($\Delta_{l=1} \approx 0.0136, \Delta_{l=2} \approx 4.716 \cdot 10^{-6}$) can be found in \cref{fig:LHVnet:maxstrat}. 
It is peculiar that there are some approximate symmetries also in this response function (Alice and Bob's symmetry under exchange of blue and red and of green and yellow, if the $\gamma$ axis is swapped, as well as the same color change when Charlie flips her strategy about the diagonal axis). 
Though not exact, such symmetries hint at some deeper structure that is not yet well understood.

Finally, note that any known local distribution can be used to obtain an upper bound on the $\delta_w$ value, by rearranging \cref{eq:ineq:basic} as
\begin{equation}
    \delta_w \leq \Delta_{l} + w \left( \frac{100}{256} - \saaa - \Delta_{l} \right)
\end{equation}
For example for $l=1$, if we evaluate this for the deterministically all-1 distribution ($p(abc)=[1,1,1]$), we find that
\begin{equation}
    \delta_w \leq \frac{3}{2}-\frac{135}{64}w \qquad(l=1),
\end{equation}
which immediately implies that only $w<0.711$ are candidate $w$ values, as $\delta_w>0$ is required.
Moreover, for $w=0.678$ this expression gives us $\delta_w\leq 0.069$, in correspondence with the results in the maintext.
Based on the numerics (\cref{fig:LHVnet:inequality_test}) we believe that for the $l=1$ case this bound (and the corresponding bound for the largest $\saaa$ symmetric distribution) will ultimately give the appropriate value of $\delta_w$. Writing this bound for the $\saaa=0.25$ distribution, we get
\begin{equation}
    \delta_w \leq \frac{9}{64} w \qquad(l=1),
\end{equation}
which is perhaps not tight, as there might be a symmetric distribution with $\saaa\approx0.289$ which would give an even tighter bound. In fact, evaluating the (not precisely symmetric) model found by LHV-Net gives 
\begin{equation}
    \delta_w \leq 0.013627  + 0.082912 \, w \qquad(l=1),
\end{equation}
giving $\delta_w\leq0.0698$ for $w=0.678$.

\subsection{Squared penalty inequalities ($l=2$)}
For the $l=2$ penalty function, there seems to be an additional intermediate regime, where neither the all-111 nor the fully symmetric $\saaa\approx 0.25$ distribution are optimal. We find that here LHV-Net actually recovers the distribution of \cref{eq:high111local}, which we name ``squares'' distribution due to the four large squares appearing in the response functions. 

We evaluate several local distributions to get the following bounds on $\delta_w$:
\begin{alignat}{3}
    & \delta_w &&\leq \frac{9}{64}w &&\qquad(l=2; \saaa=0.25),\\
    & \delta_w &&\leq 4.7163 \cdot 10^{-6} + 0.0965\, w &&\qquad(l=2; \text{LHV-Net, } \saaa=0.294, \Delta_2 = 4.7163 \cdot 10^{-6}),\\
    & \delta_w &&\leq \frac{5}{96} -\frac{124}{768}w &&\qquad(l=2; \text{``squares'' strategy, \cref{eq:high111local}})
\end{alignat}
In particular these bounds imply that only $w<\frac{10}{31} \approx 0.3226$ are viable $w$ values, and that $\delta_w<0.015$ for $w^*\approx 0.16$, which is of course a bit farther from the recovered $\delta_w^* \approx 0.012$, as can be seen from the inset in \cref{fig:LHVnet:inequality_test} (right). Moreover, notice that the $\delta_w$ bound of the fully symmetric $\saaa=0.25$ distribution and the intermediate ``squares'' distribution \cref{eq:high111local} meet at $w\approx 0.1724$, close to the numerically extracted optimal value of $w^* \approx 0.16$, hinting that the transition region between these two distributions is highly relevant for finding an inequality for $l=2$.

\begin{figure*}[t!]
    \centering
    	\includegraphics[width =0.49\textwidth]{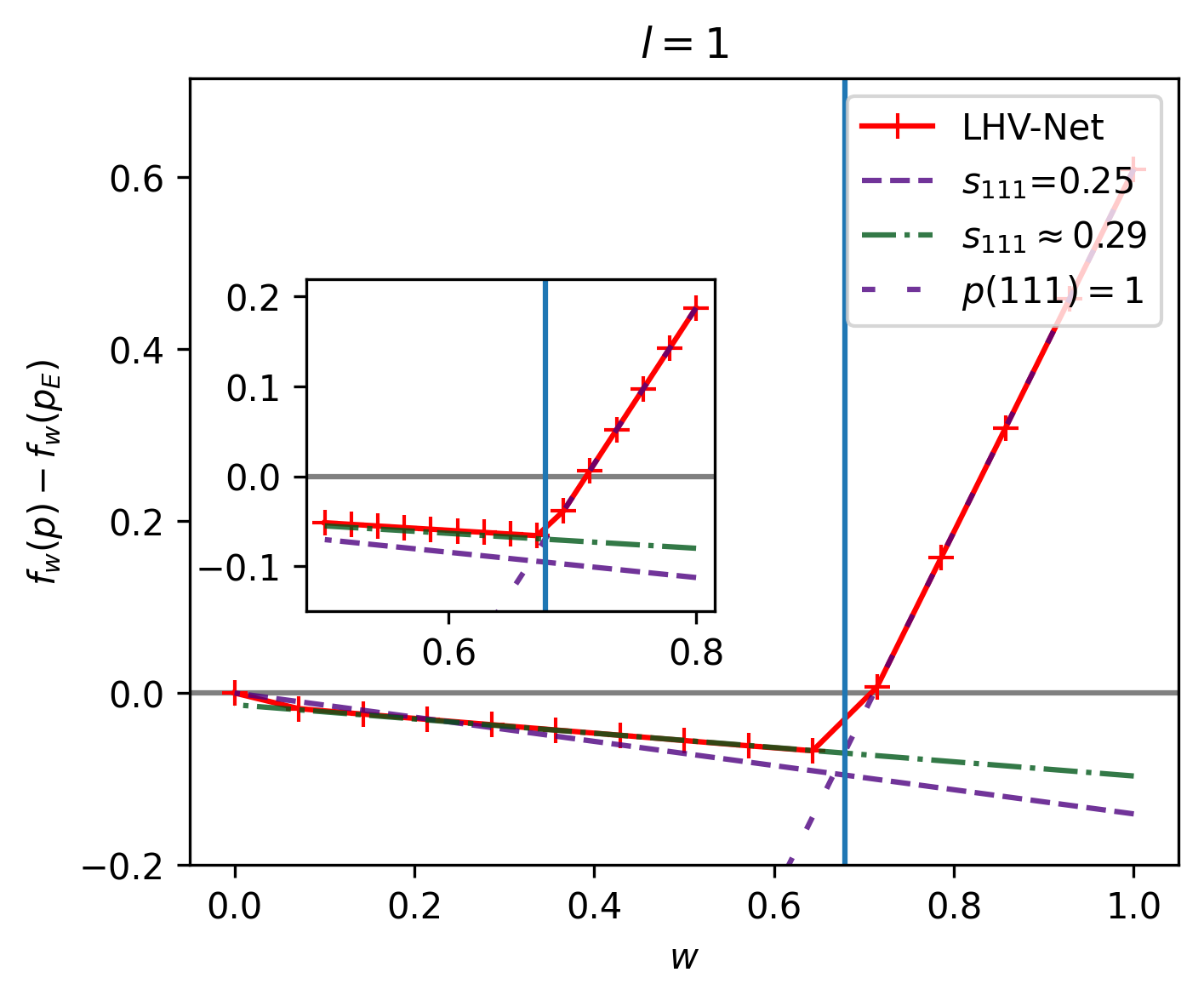}
    	\includegraphics[width =0.49\textwidth]{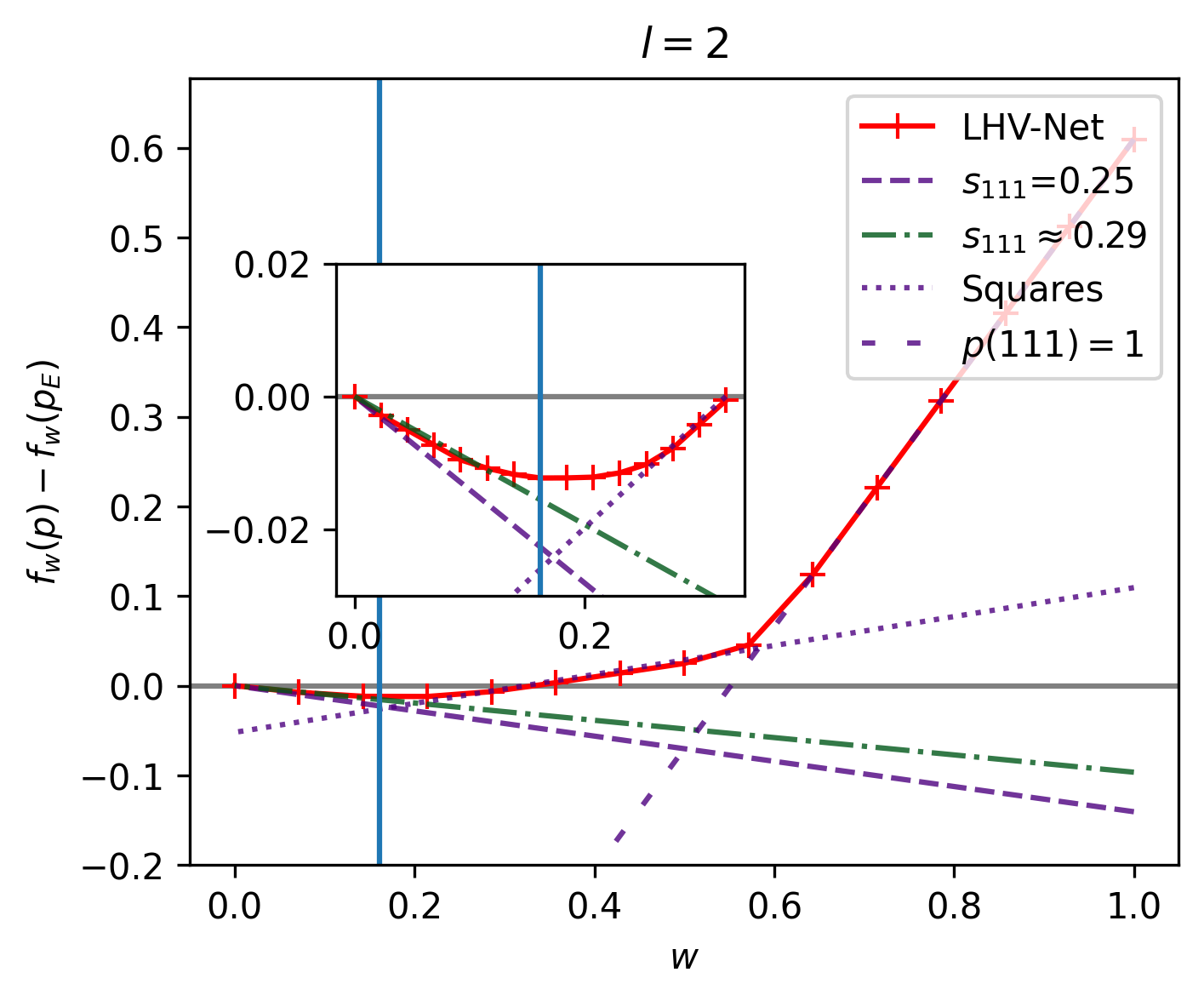}
    \caption{LHV-Net trying to maximize $f_w(p) - f_w(\pejm)$ (see \cref{eq:ineq:basic}) as a function of $w$, for $l=1$ (left) and $l=2$ (right). Additionally we plot the EJM distribution (grey, at $y=0$), deterministically outputting-1 distribution (i.e. $p(abc)=[1,1,1]$, sparsely dashed violet, highly asymmetric, but high $\saaa$), a symmetric distribution with $\saaa=\frac{1}{4}$ (dashed violet, current best analytic strategy that is fully symmetric), and the distribution found by LHV-Net that is almost symmetric and has $\saaa\approx0.29$ (green dot-dashed). Moreover, for $l=2$ we also depict the values of the distribution in \cref{eq:high111local} (called ``Squares'', dotted violet). Candidate $w$ parameters are those where the LHV-Net's results are below the EJM's value, i.e. below zero. These correspond to a positive $\delta_w$. The best $\delta_w$ estimates are signaled by vertical lines and are for $l=1$ ($l=2$) at $w^*\approx 0.678$ ($w^*\approx 0.16$) with $\delta_{w^*}\approx 0.069$ ($\delta_{w^*}\approx 0.012$).}
    \label{fig:LHVnet:inequality_test}
\end{figure*}

\subsection{$\saaa>\sfrac14$ local distribution?}
\label{app:saaa>0.25}

When maximizing the inequalities for $l=1$ for small $w$ values (i.e. in the regime symmetry is crucial), LHV-Net finds a distribution which is almost symmetric, but has $\saaa>\sfrac14$. Below are some characteristics of the distributions, as well as for a deterministic approximation of the LHV-Net's distribution (whose flags can be seen in \cref{fig:LHVnet:maxstrat}(left)), and a reference distribution which is analytically symmetric and has $s_{111}=\sfrac14$. The full LHV-Net and its deterministic approximation distributions can be found in the data appendix~\cite{data_appendix}. Moreover, we plot the distributions in \cref{fig:LHVnet:maxstrat}.

\begin{center}
\begin{tabular}{c||c|c|c|c|c}
          & $s_{111}=\sfrac14$ distr. & LHV-Net mean ($\pm$ std. dev.) & LHV-Net range & det. approx. mean ($\pm$ std. dev.) & det. approx. range\\ \hline \hline
$p(111)$ & 0.0625     &    $0.0723 \pm 0.5\%$      & $[0.0719,0.0727]$ &    $0.0735 \pm 0.2\%$      & $[0.0733,0.0738]$ \\
$p(112)$ & 0.0125     &    $0.0118 \pm 1.5\%$      & $[0.0114,0.0123]$ &    $0.0116 \pm 2.4\%$      & $[0.0112,0.0124]$ \\
$p(123)$ & 0.0125     &    $0.0119 \pm 1.2\%$     & $[0.0115,0.0122]$ &    $0.0120 \pm 2.3\%$     & $[0.0116,0.0128]$
\end{tabular}
\end{center}

Though the distribution is not exactly symmetric, several pieces of numeric evidence point towards the existence of a distribution which has $\saaa>\sfrac14$ and is fully symmetric, summarized in the following list.
\begin{itemize}
    \item The nonlocality of the Elegant distribution has previously been conjectured be robustness to noise at the source (up to $20\%$ noise) and at the detectors (up to $14\%$ noise)~\cite{krivachy_neural_2020}. The extracted noise robustness values indicate that the noisy elegant distribution is local for $\saaa = 0.289$ and $\saaa=0.286$ for the two noise models, respectively.
    \item For the 3 outcome-per-party scenario ($N=3$), there is an analytic example of a distribution for which $\saaa = 7/18 >1/3$ (see \cref{app:3outcomesExample}).
    \item The maps of the symmetric subspace as found by LHV-Net indicate a bump in the local set around the area of $p(112) \approx p(123)$, where $\saaa>1/N$ is possible. The bump is visible for the examined cases of $N=3,4,5,6$ (see \cref{fig:nn cards}).
    \item When maximizing the inequality for $l=1$ in \cref{fig:LHVnet:inequality_test} (left), the slope of the maximum values are above the line that corresponds to a $\saaa = 0.25$ maximum value (and instead line up with a hypothesis of $\saaa \approx 0.29$ as a maximum value).
\end{itemize}
Though many of these evidences rely on the numerics of LHV-Net, they have been obtained in different manners. Moreover, the explicit construction of a distribution with $\saaa>1/3$ for $N=3$ is an analytic result obtained independently from numerics. Together, these lead us to believe that the structure of the local set in the symmetric subspace is more surprising than it seems at first sight, and that it can not be perfectly characterized by a simple bound on $\saaa$.

\begin{figure*}[t!]
    \centering
    \includegraphics[width =0.47\textwidth]{nn_flags.pdf}
    \includegraphics[width =0.52\textwidth]{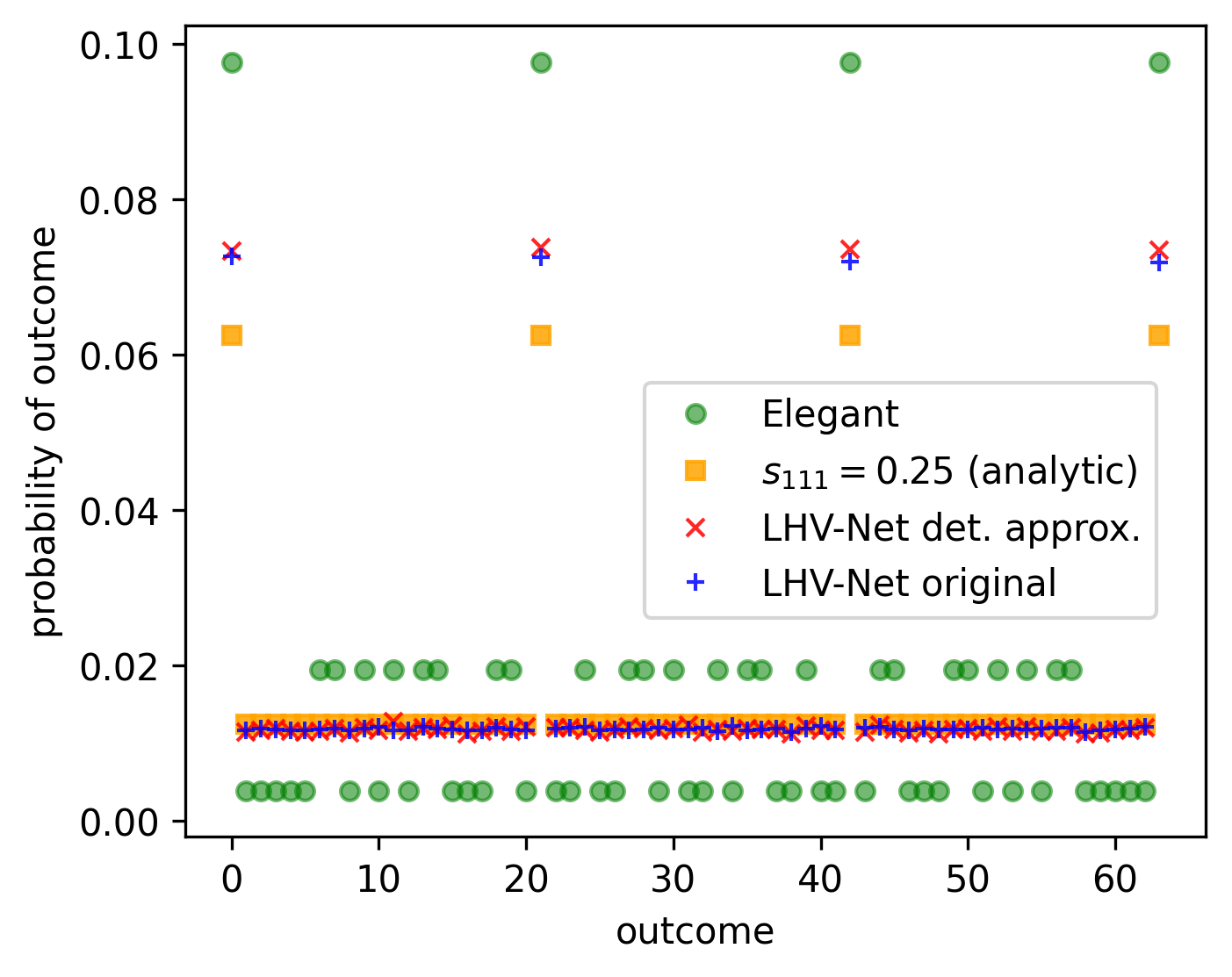}
    \caption{(Left) Flags for the deterministic approximation of the strategy found by LHV-Net, which has $\saaa\approx 0.29$. (Right) Distribution found by LHV-Net and its deterministic approximation, compared to the Elegant distribution and an analytically symmetric $\saaa=0.25$ distribution with $p(112)=p(123)$. Exact data for these flags and distributions can be found in the data appendix~\cite{data_appendix}.}
    \label{fig:LHVnet:maxstrat}
\end{figure*}

\end{document}